\documentclass[aps,twocolumn,superscriptaddress,nofootinbib,a4paper,groupedaddress]{revtex4-2}

\pdfoutput=1

\usepackage{microtype}
\usepackage{times}
\usepackage{amsfonts}
\usepackage{amsmath}
\usepackage{amsthm}
\usepackage{amsthm}
\usepackage{amssymb}
\usepackage{amsthm}
\usepackage{dsfont}
\usepackage{mathtools}
\usepackage{color}
\usepackage{multirow}
\usepackage[normalem]{ulem}
\usepackage{latexsym}
\usepackage{mathrsfs}
\usepackage{natbib}
\usepackage{verbatim}
\usepackage[T1]{fontenc}
\usepackage{float}
\usepackage{graphicx}
\usepackage{xcolor}
  \definecolor{carmine}{RGB}{150, 0, 24}
\usepackage{physics} 
\usepackage{bbold}
\usepackage{float}
\usepackage{bm}
\newcommand{\ii}{\mathrm{i}}

\usepackage{enumitem}

\usepackage{caption}
\usepackage{subcaption}
\def\half{\frac{1}{2}}

\newcommand{\A}{{\mathsf{A}}}
\newcommand{\B}{{\mathsf{B}}}

\usepackage[colorlinks=true,linkcolor=blue,citecolor=carmine,urlcolor=blue]{hyperref}

\newtheorem{theorem}{Theorem}
\newtheorem{lemma}{Lemma}
\newtheorem{definition}{Definition}
\newtheorem{conjecture}{Conjecture}

\begin{document}

\title{Classification of joint quantum measurements based on entanglement cost of localization}

\author{Jef Pauwels}
\author{Alejandro Pozas-Kerstjens}
\author{Flavio Del Santo}
\author{Nicolas Gisin}

\affiliation{Group of Applied Physics, University of Geneva, 1211 Geneva, Switzerland}
\affiliation{Constructor University, 1211 Geneva, Switzerland}

\begin{abstract}
Despite their importance in quantum theory, joint quantum measurements remain poorly understood. An intriguing conceptual and practical question is whether joint quantum measurements on separated systems can be performed without bringing them together. Remarkably, by using shared entanglement, this can be achieved perfectly when disregarding the post-measurement state. However, existing localization protocols typically require unbounded entanglement. In this work, we address the fundamental question: ``Which joint measurements can be localized with a finite amount of entanglement?'' We develop finite-resource versions of teleportation-based schemes and analytically classify all two-qubit measurements that can be localized in the first steps of these hierarchies. These include several measurements with exceptional properties and symmetries, such as the Bell state measurement and the elegant joint measurement. This leads us to propose a systematic classification of joint measurements based on entanglement cost, which we argue directly connects to the complexity of implementing those measurements. We illustrate how to numerically explore higher levels and construct generalizations to higher dimensions and multipartite settings.
\end{abstract}

\maketitle

\section{Introduction}

Joint measurements are pivotal in quantum theory, featuring prominently in key quantum information primitives such as entanglement swapping \cite{Zukowski1993}, teleportation \cite{Bennett1993}, and dense coding \cite{Bennett1992}. Despite their importance, joint measurements have been somewhat eclipsed by the extensive research on joint states. Systematic characterizations of entangled states \cite{HorodeckiReview2009}, the quantification of entanglement in states \cite{Plenio2007}, and methods for witnessing its presence with minimal assumptions \cite{Werner1989,Augusiak2014} represent just a few examples of the many problems that have been extensively studied. When it comes to measurements, especially entangled measurements—-where some or all POVM elements are entangled—-the situation is radically different. Some surprising phenomena are indeed known, such as the existence of measurements that, despite being described by orthogonal product states, have eigenstates that cannot be perfectly distinguished via local operations and classical communication \cite{Bennett1999}. However, a systematic characterization of joint measurements is lacking even in the simplest scenario of measurements on systems of two qubits \cite{DelSanto2024}. This gap underscores a vital segment of quantum information theory that merits further exploration \cite{Cavalcanti2023}. 

A deeper understanding of joint quantum measurements is also practically relevant. Establishing quantum correlations over long distances will require intermediate nodes using quantum repeaters \cite{Briegel1998}, which perform joint measurements on the systems received to distribute entanglement. Moreover, the ability to perform joint measurements is crucial for developing quantum networks \cite{networkReview}, which establish quantum correlations between multiple parties \cite{Saunders2017,Carvacho2022,Gu2023,Piveteau_2022}.

A historically important question has been whether the measurement of nonlocal observables can be localized, in the sense that the quantum-to-classical transition occurs locally, independently and instantaneously in each part of the joint system being measured \cite{Popescu1994,Sorkin1993,Gisin2024}. The original motivation for this question was to understand if the quantum measurement process is compatible with the theory of relativity \cite{Aharonov1980,Aharonov1981}. The apparent tension with relativity was already noted by Landau and Peierls in 1933 \cite{Landau1931}, who even contemplated the need for new principles to preclude the \textit{a priori} existence of nonlocal observables. This issue remains an active area of inquiry today, particularly in the context of quantum field theory \cite{Fewster2023}, following the foundational work of Sorkin \cite{Sorkin1993}. The seminal results by Vaidman, Reznik, and Groisman \cite{Groisman2001,Groisman2003,Vaidman2003} show that every joint measurement can be implemented by isolated parties when these parties have access to a source of shared entangled particles. This has led to many exciting applications, such as schemes for nonlocal quantum computation and position-based quantum cryptography \cite{Beigi2011,Gonzales2020,Buhrman2014}. However, these protocols \cite{Vaidman2003,Clark2010,Ishizaka2008} typically require an infinite amount of entanglement, either on average \cite{Vaidman2003,Ishizaka2008} or in the worst case \cite{Clark2010}.

\begin{figure}[t]
  \centering
  \includegraphics{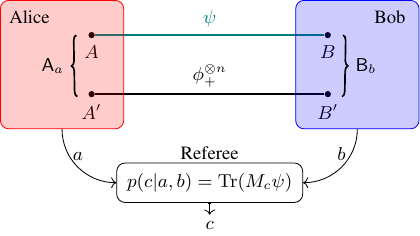}
  \caption{A joint measurement $M_{AB} = \{M_c\}_c$ is $n$-ebit localizable if and only if the Born statistics $p(c|M_c,\psi)$ for any state $\psi$ can be reconstructed from the outcomes $a$ and $b$ of Alice and Bob's local operations $\A_a$ and $\B_b$ on their respective subsystems $A$ and $B$, with the assistance of $n$ copies of the maximally entangled state $\phi_+$ (see Definition~\ref{def:localizable}).}
  \label{fig:scenario}
\end{figure}

In this work, we ask the fundamental question: ``Which joint measurements can be performed locally with a \emph{given amount of entanglement}?'' (see Figure~\ref{fig:scenario} and Definition~\ref{def:localizable}). To this end, we develop a finite-resource version of the teleportation-based protocol described in Ref. \cite{Vaidman2003}. This protocol is closely connected to the Clifford hierarchy \cite{Gottesman1999}, which is in turn connected to well-established notions of circuit complexity \cite{Buhrman2013,Speelman2016,Chakraborty2015}. This leads us to propose entanglement cost of localization as the basis for a classification of joint measurements. We analytically characterize all two-qubit measurements that can be localized at the first levels of the hierarchy. These measurements are particularly symmetric and correspond to several known measurements with interesting properties, in particular the complete and partial Bell state measurements, the elegant joint measurement \cite{Gisin2019} and (twisted) Bell state measurements \cite{Aharonov1980,Boreiri2023}, consistent with the intuition that the low levels of our proposed classification correspond to measurements of low complexity. Our work represents a first step in the systematic characterization of joint quantum measurements, which we discuss how to generalize to more higher levels, higher-dimensional systems, and multiple parties.

\section{Preliminaries}\label{sec:formalism}
Can a measurement on a joint system $AB$ be conducted while only locally interacting with each subsystem $A$ and $B$ within separate isolated laboratories? Let us consider a measurement, represented by a projective operator-valued measure (POVM)  $M_{AB}=\{M_{c}\}_c$. The measurement process in quantum theory is governed by two postulates: (i) the distribution over the measurement outcomes $c$ is given by Born's rule, $p(c) = \Tr(M_c \psi)$, where $\psi$ is the state of the quantum system to be measured, and (ii) the state of the system after observing the outcome $c$ is given by L\"uders' rule; for projective measurements, $\psi_c = M_c \psi M_c / p(c)$.
In this work we focus on what is actually observable in an experiment: the statistical aspect of the measurement, disregarding the post-measurement state. Indeed, if one insists on reproducing the post-measurement state according to L\"uders' rule, it has been established that only measurements that erase all local information in their state-update rule can be performed locally without breaking causality \cite{Popescu1994}. 

Our question can thus be phrased in terms of a correlation game, see Fig.~\ref{fig:scenario}. The game features three parties: Alice, Bob, and an external Referee. Alice and Bob want to prove to the Referee their ability to simulate the statistics of a joint measurement $\{M_c\}_c $. The Referee prepares a joint state, $\psi_{AB}$, and gives the subsystems $A$ and $B$ to respectively Alice and Bob. Each party can perform arbitrary local operations on their subsystems, but they are not allowed to communicate. Next, each party relays only classical data to the Referee. If the Referee can reconstruct the statistics of the measurement $\left\{p(c) = \Tr(M_c \psi)\right\}_c$ from the combined classical data, the parties win the game. We say that the measurement $M_{AB}$ is \emph{localizable} if the parties can win the game, with unit probability, regardless of the input state.

Clearly, the only measurements which can be localized according to the scheme above --without shared resources-- are product measurements, since implementing (classically or quantumly) correlated measurements requires of a source of prior (classical or quantum) correlations between the parties.
The situation thus changes when the parties have access to shared entanglement \cite{Aharonov1980}. Indeed, it was shown by Vaidman and others that, with the assistance of entanglement in the form of EPR pairs (also called \textit{ebits}) shared between Alice and Bob, every measurement can be localized \cite{Vaidman2003}.
However,  in order to run the scheme proposed in Ref.~\cite{Vaidman2003} one \emph{always} requires an infinite amount of entanglement to localize any measurement.
These results were refined in Ref.~\cite{Clark2010}, where the authors proposed a protocol that, despite still requiring infinite entanglement in the worst case, consumes a finite number of ebits on average.

The central question addressed in this work is complementary, namely: given a finite amount of entanglement, which joint measurements can Alice and Bob localize? To this end, we quantify entanglement by the number of EPR pairs, and formalize the following notion of localizability:
\begin{definition}[$n$-ebit localizable]
    A quantum measurement $M_{AB}$ represented by a POVM $\{M_c\}_c$, acting on a separable Hilbert space $\mathcal{H} = \mathcal{H}_A \otimes \mathcal{H}_B$, is $n$-ebit localizable if and only if there exist local measurements $\{\A_a\}_a \subset \mathcal{H}_{AA'}$ and $\{\B_b\}_b \subset \mathcal{H}_{BB'}$, where $A'$ and $B'$ are local ancillas, such that for all $c$ the following holds:
    \begin{equation} \label{eq:generaldefinition}
        M_c = \sum_{a,b} p(c|a,b) \Tr_{A'B'}\left[\left(\A_a \otimes \B_b\right) \left(\mathbb{1}_{AB}\otimes{\phi^+}^{\otimes n}_{A'B'} \right)\right].
    \end{equation}
    \label{def:localizable}
\end{definition}
We note that the above definition is equivalent to the requirement that the Born statistics can be reproduced for any input state $\psi$.

In this work, we restrict ourselves to projective non-degenerate measurements\footnote{Any degenerate measurement can be obtained trivially by binning the outcomes of a full-rank measurement.}. The former simplifies the problem since linearity of Born's rule implies that it is sufficient to check Definition \ref{def:localizable} for the eigenstates of the measurement.
Such measurements can be represented as a unitary matrix, $M$, where each column  is an eigenvector, $\ket{\psi_i}$, of the measurement.
For example, for a two-qubit measurement, the matrix $M$ is given by
\begin{equation}
  M=\left( \ket{\psi_1}, \ket{\psi_2}, \ket{\psi_3}, \ket{\psi_4}\right),
  \label{eq:M}
\end{equation}
where $\ket{\psi_i} \in \mathbb{C}^4$ are column vectors.
Below, we simply refer to $M$ as the measurement for brevity.
Note that $M^\dagger$ is also a unitary matrix, that diagonalizes the measurement. This is, in order to perform the measurement, one applies $M^\dagger$ to any state and measures the resulting state in the computational basis.

For the sake of localizablity, we are interested only in equivalence classes of measurements. Indeed, if a measurement is localizable, then so is any local unitary rotation of it. Similarly, if a measurement $M_{AB}$ is localizable, then so is the measurement $M_{BA}$, in which its subsystems are permuted\footnote{This is relevant because, as will become clear later, localization protocols are often asymmetric between the parties.}. 
We can also permute the basis states, which corresponds to a relabelling of the outcomes, and change the global phases of each eigenvector in the unitary representation $M$, which is a physically irrelevant artifact of the unitary representation. This leads us to the following definition:

\begin{definition}[Equivalent measurements] \label{def:equivalence}
  Two measurements $M_{AB}$ and $M'_{AB}$ are equivalent $M_{AB} \sim M'_{AB}$ if and only if there exist local unitaries, $U_A$ and $U_B$, such that $M'_{AB} = M_{\pi(AB)} \cdot \left(U_A \otimes U_B\right) \cdot \tilde{P} \cdot \Phi $, where $\tilde{P}$ is a permutation matrix, $\pi$ denotes a permutation of the subsystems $A$ and $B$, and $\Phi$ is a diagonal matrix with phases.
\end{definition}

\section{Localization protocols}

The usual way to perform a joint measurement is to simply bring the subsystems together. Equivalently, this can be achieved by one party teleporting their part of the joint system to the other party, who then performs the measurement. Neither of these strategies is allowed in the scenario we consider, because any form of communication, quantum or classical, is disallowed. Nevertheless, the teleportation protocol is the inspiration for a class of \textit{blind ping-pong teleportation} protocols, which include the protocol by Vaidman \cite{Vaidman2003} as well as its modification by Clark \textit{et al.} \cite{Clark2010}. The basic idea of such protocols is that Alice and Bob teleport the state of the system they want to measure back and forth blindly (e.g., without communicating the correction) between their laboratories, each applying local operations and local Bell measurements at each step. Each teleportation introduces Pauli distortions unknown to the other party, and the protocols continue until these are corrected.

Alternative protocols have been explored, including some based on port-based teleportation \cite{Ishizaka2008}. Although these protocols still require infinite resources to localize arbitrary measurements without errors and with unit probability, they improve the scaling of entanglement consumption from doubly exponential to exponential in both deterministic and probabilistic versions \cite{Beigi2011}. It remains uncertain whether infinite entanglement is also needed if the requirement of universality is relaxed (i.e., if the protocol is only required to work for specific measurements). This relaxation would circumvent the no-programming theorem \cite{Nielsen1997}, which prohibits any universal protocol for port-based teleportation from working deterministically and with unit fidelity. Another class of protocols is adaptive protocols, in which the parties adjust their local strategies based on the outcomes of previous measurements. These protocols are known to work deterministically for specific measurements \cite{Groisman2003,Clark2010,Gisin2024}. However, this approach necessarily depends on the specifics of the measurement, unlike blind-teleportation protocols, which work universally. 

Below, we start by sketching the general idea of the Vaidman protocol \cite{Vaidman2003}\footnote{See also Ref.~\cite{Clark2010} for an excellent explanation.} before presenting our version with finite resources.

\subsection{Vaidman protocol}
The scheme starts by Bob teleporting his part of the shared system $\ket{\psi}$ by performing a Bell state measurement (BSM) to his part of $\ket{\psi}$ and a shared ebit. 
Alice then holds the state $\mathbb{1} \otimes \sigma_{b_0} \ket{\psi}$, where $b_0=0, X,Y,Z$ denotes the outcome of Bob's BSM (unknown to Alice) and we define $\sigma_0 = \mathbb{1}$. She attempts to transform the initial state into the computational basis and teleport it back, so that Bob can measure in the computational basis to complete the measurement.
To this end, she applies the unitary $M^{-1}=M^\dagger$ (where, recall, $M$ is the matrix, given by Eq.~\eqref{eq:M}, whose adjoint rotates the target measurement basis to the computational basis), and uses two more ebits to teleport the resulting two-qubit system back to Bob.
Bob then holds the state $\left(\sigma_{a_0} \otimes \sigma_{a_1}\right) \cdot M^\dagger \cdot \left(\mathbb{1}\otimes\sigma_{b_0}\right) \ket{\psi}$, where $a_0$ and $a_1$ are the outcomes of Alice's BSMs.

If Bob's initial BSM was non-distorting (i.e., $b_0=0$, for which $\sigma_0=\mathbb{1}$), he measures in the computational basis and the measurement is completed. This is so because the subsequent distortions, $\sigma_{a_0} \otimes \sigma_{a_1}$, simply amount to a relabelling of the outcomes, which can be corrected by the Referee.
If $b_0=X,Y,Z$, the protocol takes on a tree-like structure, where each level has a branch for each possible distortion produced by Bob's measurement.
Every branch is labelled by $b_0$ and contains two ebits. 
Bob uses the $b_0$-th branch to teleport the system back to Alice, who receives the state $\left(\sigma_{b_1}\otimes\sigma_{b_2}\right)\cdot\left(\sigma_{a_0}\otimes\sigma_{a_1}\right)\cdot M^\dagger\cdot\left(\mathbb{1}\otimes\sigma_{b_0}\right)\ket{\psi}$ in that branch.

Alice does not know which branch was used by Bob, but knows how each branch is distorted. Thus, she applies to each branch the corresponding correction $U_{b_0}$ that satisfies
\begin{equation}
  \left(U_{b_0}\otimes\mathbb{1}\right)\cdot\left(\sigma_{a_0}\otimes\sigma_{a_1}\right)\cdot M^\dagger\cdot\left(\mathbb{1}\otimes\sigma_{b_0}\right) = M^\dagger
\end{equation}
i.e., that would produce $M^\dagger$ if Bob's teleportation were non-distorting,
and then teleports each branch back to Bob.

If $b_1=b_2=0$ (which happens with probability $1/16$), the state received by Bob in the relevant branch is $\left(\sigma_{a_2}\otimes\sigma_{a_3}\right)M^\dagger\ket{\psi}$.
Thus, he can measure in the computational basis and send his result to the Referee in order to compute the final outcome of the measurement defined by $M$ in the state $\ket{\psi}$.
In every other case, the correction $U_{b_0}$ gets distorted by Pauli matrices. The protocol continues in an analogous way, with Bob teleporting the state back to Alice in one of $3\times15$ branches, who applies a correction $U_{b_0,b_1,b_2}$ for each of them, and so on.
The process is repeated until Bob's BSMs are non-distorting.
Importantly, when this happens is only known to Bob, so Alice has to continue the scheme forever, which implies that the entanglement consumption of the protocol is always infinite.

\begin{figure}[ht]
  \centering
  \begin{subfigure}[One ebit]{0.95\columnwidth}
      \centering
      \includegraphics[width=0.95\columnwidth]{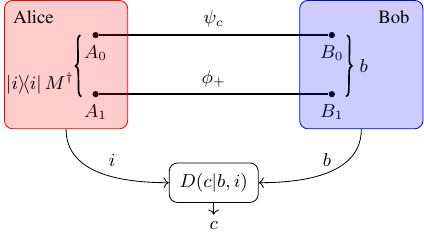}
      \caption{\label{fig:1ebit}
        Finite Vaidman protocol with one auxiliary ebit.}
  \end{subfigure}
  \begin{subfigure}[Three ebits]{0.95\columnwidth}
      \centering
      \includegraphics[width=0.95\columnwidth]{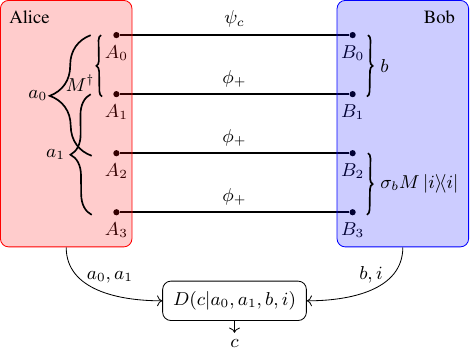}
      \caption{\label{fig:3ebits}
      Finite Vaidman protocol with three auxiliary ebits.}
  \end{subfigure}
      
  \caption{The first (\protect\subref{fig:1ebit}) and  second (\protect\subref{fig:3ebits}) level of the finite-consumption version of the Vaidman localization protocol. The classical labels $a_0$, $a_1$ and $b_0$ denote the outcomes of Bell state measurements performed in the local lab of respectively Alice and Bob. $\ketbra{i}$ denotes a measurement in the computational basis. If the measurement $M$ is localizable, there exists a deterministic decoding $D$ which identifies all eigenstates of the measurement $\psi_c$ from the local data $i,a_i,b$ of Alice and Bob.}
  \label{fig:blind_teleportation}
\end{figure}

Clark \textit{et al.} \cite{Clark2010} introduced a modification of this scheme, that consumes a finite amount of entanglement on average (but still requires infinite entanglement in the worst case).
The basic idea is that instead of applying the unitary $M$ directly, the parties decompose it into a series of Pauli rotations $M = \Pi_{\bm{k}} R_{\bm{k}}(\theta_k)$, where $R_{\bm{k}}(\theta) =\exp(\ii \theta \sigma_{k_1} \otimes \sigma_{k_2})$. The crucial advantage of this is that if a teleportation-induced Pauli distortion does not commute with the Pauli rotation, it always leads to the application of a rotation in the inverse direction, which can be corrected by applying a double rotation. Each party can thus have a termination criterion, which yields the average entanglement cost finite. In Appendix~\ref{app:clark} we explain this scheme in more detail and give details on the conditions under which it succeeds with a finite amount of entanglement.

\subsection{Finite-consumption version}
Here we develop a finite-consumption version of the Vaidman scheme described above, in which after some teleportation step, the party holding the teleported state aborts the protocol by measuring in the computational basis. In general, this will be successful only if the previous teleportations were non-distorting (i.e. probabilistic). However, as we shall see, for some measurements which satisfy certain symmetries, the protocol succeeds with unit probability, regardless of eventual distortions.
A version of this protocol has been considered previously in the context of position-based cryptography \cite{Chakraborty2015}.

We begin with the first level of the protocol, illustrated in Fig.~\ref{fig:1ebit}.
Interestingly, this protocol was already introduced long before the general protocol \cite{Aharonov1980,Aharonov1981}.
Bob starts by performing a BSM on his share of the system $\psi$ and his part of a shared $\phi^+$.
This teleports the state to Alice, who now holds $ \mathbb{1} \otimes \sigma_b \ket{\psi}$, where $b$ is the outcome of Bob's BSM. 
Next, she applies $M^\dagger$ to the joint system and measures in the computational basis.
If Alice's measurement statistics match those of the original measurement, up to a deterministic post-processing depending on Bob's local outcomes, the measurement can be localized using one ebit.

The Referee can reconstruct the correct measurement statistics if and only if $\left|\bra{\pi_b (i,j)} M^\dagger  \cdot \left(\mathbb{1} \otimes \sigma_b\right) \ket{\psi}\right|^2 = \left|\bra{i,j} M^\dagger \ket{\psi}\right|^2\,\forall\,b$, where $\pi_b$ represents an arbitrary permutation of the computational basis states dependent on $b$.
Thus, $M$ is localizable with one ebit if there exist permutation matrices, $\tilde{P}_b$, and diagonal matrices of phases, $\Phi_b$, (since the absolute value is used to compute the probabilities) such that
\begin{equation} \label{eq:1ebit}
M^\dagger\cdot\left(\mathbb{1}\otimes\sigma_b\right)\cdot M =  \tilde{P}_b \cdot \Phi_b \equiv P_b
\end{equation}
for all $b$.
We note that the equation $b=0$ simply imposes unitarity. 
From the three remaining ones, one is redundant since it can be obtained by multiplying the other two, e.g., $b=X$ and $b=Z$. 

The next level of the protocol uses three ebits, see Fig.~\ref{fig:3ebits}. After applying $M^\dagger$ Alice, instead of measuring in the computational basis, uses the two remaining ebits to teleport both her qubits back to Bob.
Since Bob knows $b$ and thus which distortion $\sigma_{b}$ was applied to $\ket{\psi}$, he tries to correct it by applying the unitary $M^\dagger \cdot \left(\mathbb{1} \otimes \sigma_b\right)\cdot M$, corresponding to the operation he would perform in order to produce the state $M^\dagger\ket{\psi}$ if Alice's teleportation were non-distorting. Allowing for the same freedoms as in the first level (recall, relabeling of the computational basis and arbitrary phases), we find that a measurement $M$ can be localized with three ebits if it satisfies the following set of $4^3$ equations:
\begin{equation}
  M^\dagger \cdot \left(\mathbb{1} \otimes \sigma_b\right)\cdot M \cdot\left(\sigma_{a_1} \otimes \sigma_{a_2}\right)\cdot M^\dagger \cdot\left(\mathbb{1} \otimes \sigma_b\right)\cdot M = P_{a_1,a_2,b},
  \label{eq:3ebit}
\end{equation}
for $a_1$, $a_2$, $b =0,X,Y,Z$ and where each $P_{a_1,a_2,b}$ is a permutation matrix with free phases.
Note that, in terms of the matrices $\mathcal{M}_b=M^\dagger\cdot\left(\mathbb{1}\otimes\sigma_b\right)\cdot M$, Eq.~\eqref{eq:3ebit} takes the same form as the one-ebit equation, Eq.~\eqref{eq:1ebit}.

From here on, the protocol takes on the tree-like structure described earlier. Thus, the next level requires nine ebits (the one used in the first level, the additional two used in the second level, and two more for each possible distortion at the first step by Bob, $\sigma_{b}$ with $b\ne 0$). 
Instead of measuring at the last step of the second level, Bob uses the two ebits in the branch labeled by $b$ to teleport the system back to Alice.
Alice does not know which of the branches was used by Bob, so in each branch she applies the appropriate correction corresponding to correcting her earlier distortions plus the distortion induced by Bob, which she knows from the label of the branch.
After that, she measures all her systems in the computational basis.
Measurements for which, independent of the distortion (i.e., the branch), this correction is successful are said to be localizable with nine ebits and satisfy
\begin{equation}
  \mathcal{M}_{a_1,a_2,b}^\dagger \cdot \left(\sigma_{b_1}\otimes\sigma_{b_2}\right)\cdot\mathcal{M}_{a_1,a_2,b} = P_{a_1,a_2,b_1,b_2,b},
  \label{eq:9ebit}
\end{equation}
where $\mathcal{M}_{a_1,a_2,b}=
\mathcal{M}^\dagger_b\cdot\left(\sigma_{a_1}\otimes\sigma_{a_2}\right)\cdot\mathcal{M}_b$ with $\mathcal{M}_b=M^\dagger\cdot\left(\mathbb{1}\otimes\sigma_b\right)\cdot M$, and $P_{a_1,a_2,b_1,b_2,b}$ is, again, a permutation matrix with free phases for every $a_1$, $a_2$, $b_1$, $b_2$ and $b$.

This scheme can be continued indefinitely, with the number of ebits required growing double-exponentially at each step.
Note that, given the form of Eqs.~\eqref{eq:1ebit}-\eqref{eq:9ebit}, the formal generalization of the relevant equations is straightforward, and each corresponds to a set of nested one-ebit-like equations. 

\section{Connection to the Clifford hierarchy and complexity}

The hierarchy induced by the set of conditions for localizability in the finite Vaidman scheme (Eqs.~\eqref{eq:1ebit}-\eqref{eq:9ebit} and their generalization) bear strong similarity to the Clifford hierarchy, introduced in Ref.~\cite{Gottesman1999} in the context of teleportation-based computation.

The $k$-th level of the two-qubit Clifford hierarchy is defined as the set of stabilizers, $\mathcal{C}_k$, of the elements in the previous level, $\mathcal{C}_{k-1}$, where the zeroth level is defined to be the Pauli group, i.e., the set of all tensor products of Pauli matrices. The hierarchy is thus defined recursively as
\begin{equation}
  \mathcal{C}_k \equiv \{U \in \mathcal{U}(4) \,|\, U^\dagger \cdot \mathcal{P}_2 \cdot U \in \mathcal{C}_{k-1} \,\} \,,
  \label{eq:clifford}
\end{equation}
where $\mathcal{C}_0 \equiv \mathcal{P}_2 = \mathcal{P}_1^{\otimes 2}$ and $\mathcal{P}_1\equiv \{\openone,\sigma_X,\sigma_Y,\sigma_Z \}$.
In contrast, our hierarchy (Eqs.~\eqref{eq:1ebit}-\eqref{eq:9ebit} and generalizations) can be defined as
\begin{equation}
  \mathcal{V}_k \equiv \{M \in \mathcal{U}(4) \,|\, M^\dagger \cdot \left(\openone \otimes \mathcal{P}
  _1\right) \cdot M \in \bar{\mathcal{V}}_{k-1} \} \,,
  \label{eq:VaidmanHierarchy}
\end{equation}
where the set $\bar{\mathcal{V}}_k$ is described recursively via
\begin{equation}
  \bar{\mathcal{V}}_k \equiv \{M \in \mathcal{U}(4) \,|\, M^\dagger \cdot \mathcal{P}_2 \cdot M \in \bar{\mathcal{V}}_{k-1} \}
  \label{eq:VaidmanHierarchy2}
\end{equation}
and $\mathcal{V}_0=\bar{\mathcal{V}}_0 \equiv \{\tilde{P} \cdot \Phi\}$ is the set of permutations with phases. Note that Eqs.~\eqref{eq:clifford} and \eqref{eq:VaidmanHierarchy2} are the same, so the differences between $\mathcal{C}_k$ and $\mathcal{V}_k$ lie in the initial and final steps of their constructions.

Comparing these definitions, we see that every measurement $M$ in the $k$-th level of the Clifford hierarchy is also an element of the finite-consumption Vaidman hierarchy at the same level, $\mathcal{C}_k \subset \mathcal{V}_k$, but not vice versa. Indeed, the defining relations of $\mathcal{V}_k$, Eqs.~\eqref{eq:VaidmanHierarchy}, \eqref{eq:VaidmanHierarchy2}, are strictly weaker than those of Eq.~\eqref{eq:clifford}. First, at the zeroth level, we already have that $\mathcal{C}_0 \subset \mathcal{V}_0$ but $\mathcal{V}_0 \not\subset \mathcal{C}_0$, since all Pauli strings of length 2 are permutations with phases but not vice versa\footnote{This difference can be traced back to the fact that for measurements, we are only interested in probabilities, i.e. we take the absolute value squared.}. Second, we do not require that $M\in \mathcal{V}_k$ maps \emph{all} Pauli strings in $\mathcal{P}_2$ to an element of $\bar{\mathcal{V}}_{k-1}$, but only a \emph{subset} of them, those of the form $\mathbb{1}\otimes \mathcal{P}_1 \subset \mathcal{P}_2$.

This additional freedom leads to a difference already at the first level of the hierarchy. For example, the twisted basis measurement, that will be introduced in Eq.~\eqref{eq:twisted} of the next section, is part of the first level of our hierarchy but it is not an element of $\mathcal{C}_1$ since it does not transform $\sigma_X\otimes\sigma_Y$ into an object of the form $\sigma^{(1)}\otimes\sigma^{(2)}$.

The position of a unitary in the Clifford hierarchy is known to be directly connected to the complexity of implementing it \cite{Gottesman1999}. In our context, this can be interpreted as meaning that \emph{the complexity of implementing a nonlocal measurement is related to the entanglement cost of localizing it}. This connection, between the finite-consumption Vaidman scheme and the Clifford hierarchy, has been used in previous works \cite{Chakraborty2015} on position-based cryptography to prove that the complexity of implementing a nonlocal unitary (the honest verification strategy) is related to the entanglement-cost of simulating the nonlocal unitary (the cheating strategy). 

Similarly, Ref.~\cite{Speelman2016} developed a blind ping-pong teleportation protocol akin to the scheme of Clark \textit{et al.}\cite{Clark2010}, linking the amount of entanglement required for localization to the number of $T$ gates\footnote{The $T$ gate is defined as $T=\begin{pmatrix} 1 & 0 \\ 0 & e^{\ii\pi/4} \end{pmatrix}$.} used alongside Clifford gates in order to implement the measurement in the circuit model. This quantity, known as the $T$-depth, is a measure of circuit complexity \cite{Braviy2016}.

These connections to complexity theory strongly support our proposal for the entanglement cost of localization as a natural and sound basis for classification of quantum measurements. In accordance with this intuition, as we will demonstrate in the next section, measurements at the lower levels of our classification correspond to “simple” measurements, which frequently appear in various contexts of quantum information theory and have been experimentally realized in many instances.

\section{Ping-pong-localizable measurements up to three ebits} \label{sec:classification}
In this section we describe all measurements that can be localized with up to three ebits using the finite-consumption versions of the Vaidman scheme described above and the Clark \textit{et al.} scheme described in Appendix~\ref{app:clark}. 
We give the complete list of equivalence classes, and discuss their properties whenever notable.
As we will see, despite the fact that the Clark \textit{et al.} scheme is, a priori, more appealing since it consumes a finite amount of entanglement on average, the finite-entanglement version of the Vaidman scheme allows localizing a richer variety of measurements.

In Appendix~\ref{app:vaidman} we present a systematic, analytic method to find all equivalence classes of measurements that can be localized at any finite step of our finite-consumption version of the Vaidman scheme. We do this explicitly for the first two steps of the protocol by systematically solving Equations \eqref{eq:1ebit} and \eqref{eq:3ebit}. For higher number of ebits, the method still applies but becomes computationally intractable.

Our method relies on the following observations:
First, finding the set of solutions to the $k$-th level of the scheme amounts to recursively solving $k$ sets of one-ebit-like equations.
This can be seen explicitly by comparing Eqs.~\eqref{eq:1ebit}, \eqref{eq:3ebit} and \eqref{eq:9ebit}.
Essentially, the one-ebit equation tells us that localizable $M$s correspond to possible intertwiners between different representations of the Pauli group. Hence, each representation of the Pauli group with the desired structure (permutation + phases), corresponds to an equivalence class of localizable measurements. 
Once we fix a representation of the group, the one-ebit equation \eqref{eq:1ebit} is a Sylvester-type equation \cite{Sylvester1884,Bhatia1997}, which can be solved efficiently. At higher levels of the scheme, the logic is the same: we recursively solve for the intertwiners between the Pauli group and unitaries which are the solutions of the previous level of the scheme. 

The complete algorithm is quite subtle, and we defer its discussion to Appendix~\ref{app:vaidman}. We also provide computer codes, written in \textsc{Mathematica}, that implement it \cite{compapp}.

Our results are summarized by the following theorems:

\begin{theorem}[Localizable measurements with one ebit]
    There are exactly three equivalence classes of measurements that can be performed locally with one ebit in a blind teleportation protocol. Trivially, this includes the product measurement. The other two classes are the Bell state measurement and the $\pi/2$ twisted basis.
\end{theorem}

\begin{theorem}[Localizable measurements with three ebits]
    There are exactly eight equivalence classes of measurements that can be performed locally with three ebits. Aside from the three classes identified in the first level, the following classes are added: the partial Bell state measurement, the elegant joint measurement, the $E_2$ measurement, the twisted Bell state measurement (tBSM), and the $B_2$ measurement.
\end{theorem}

Theorem 1 proves Conjecture 1 in Ref.~\cite{Gisin2024}. We furthermore note that all bases localizable with up to three ebits are related to the elegant joint measurement (EJM), defined in Ref.~\cite{Gisin2019}, by the relation $M_{\rm EJM}^\dagger \sim M^\dagger$. This is, for some ordering of the columns, the \emph{adjoint} of all bases localizable with up to three ebits is equivalent to the adjoint of the EJM up to the equivalence relations of Def.~\ref{def:equivalence}. This affirms and makes more precise Conjecture 2 in Ref.~\cite{Gisin2024}, which stipulated that all measurements at this level are somehow related to the EJM. 
Below we give one representative of each equivalence class and discuss their properties. 

\begin{itemize}

\item \emph{$\frac{\pi}{2}$-twisted basis measurement.} [1 ebit]
\begin{equation}
  \begin{aligned}\frac{\pi}{2}-\text{Twisted} &= \{\ket{00},\ket{01},\ket{1+},\ket{1-} \}\\
    &=\text{CRot}_{\frac{\pi}{2}}\{\ket{00},\ket{01},\ket{10},\ket{11} \}\, ,
  \end{aligned}\label{eq:twisted}\end{equation}
where $\text{CRot}_{\theta}$ is a controlled rotation of angle $\theta/2$ along the Y axis.
Also called the BB84 basis, it is the paradigmatic example of a basis composed of product states that is not a product basis \cite{Lau2011,Chitambar2014}. This basis can also be implemented through an adaptive scheme (i.e. not in the class of ping-ping teleportation protocols) using one ebit \cite{Aharonov1980,Gisin2024}. As noted before, it is not a Clifford operation. While it appears in the first level of the hierarchy, it only appears in the second level of the Clifford hierarchy, as can be easily verified.

\item \emph{Bell state measurement (BSM)}. [1 ebit]
\begin{equation} \text{BSM} = \{\ket{\phi_\pm},\ket{\psi_\pm} \} \,, \label{eq:BSM} \end{equation} where $\ket{\phi_\pm}=(\ket{00}\pm\ket{11})/\sqrt{2},\ket{\psi_\pm}=(\ket{01}\pm\ket{10})/\sqrt{2}$.

This can intuitively easily be understood, as every Pauli distortion maps a Bell state to another Bell state, and thus any Pauli distortion amounts to a relabeling of the outcomes that can be corrected by the Referee. This is the only measurement that can be localized with one ebit that is not a product measurement, and is the only non-product measurement that can be localized ideally \cite{Popescu1994}\footnote{Indeed, Bob can prepare the post-measurement state simply by locally applying $\sigma_b$ to another shared ebit.}.

\item \emph{Elegant joint measurement (EJM).} [3 ebits]
\begin{equation}
  M_{\rm EJM}=\frac{1}{\sqrt{8}}\begin{pmatrix}
    1+\ii & -1+\ii & 1-\ii & -1-\ii  \cr
    -2\ii & 0 & 0 & -2\ii  \cr
    0 & 2\ii & 2\ii &  0 \cr
    1-\ii & -1-\ii & 1+\ii & -1+\ii
  \end{pmatrix} \,. \label{eq:EJM}
\end{equation}
This is an iso-entangled basis with tangle\footnote{The tangle of a bipartite state $\rho_{AB}$ is defined as $t=2\left(1-\text{Tr}\rho_A^2\right)$, where $\rho_A=\text{Tr}_B\rho_{AB}$ is the partial trace of $\rho_{AB}$ over subsystem $B$. It ranges from $0$ for separable states to $1$ for maximally entangled states.} $t=\frac{1}{4}$.
This measurement was introduced in Ref.~\cite{Gisin2019} for its highly symmetric properties. In particular, it generates the elegant distribution in the triangle network, which was conjectured not to admit a local hidden variable model in the triangle network. This conjecture has been intensely studied \cite{Fraser2018,Krivachy2019,Baumer2024}, and recently proven \cite{Gitton2024}.

\item \emph{The $E_2$ measurement.} [3 ebits]
\begin{equation}
 E_2 = \{(\ket{\psi^-}\pm\ket{00})/\sqrt{2}, (\ket{\psi^+}\pm\ket{11})/\sqrt{2}\}\,. \label{eq:E2}
\end{equation}
This measurement is also iso-entangled with tangle $t=\frac{1}{4}$. Like the EJM, it belongs to the Elegant family of two-qubit iso-entangled measurement family defined in Ref.~\cite{DelSanto2024}.

\item \emph{Partial Bell state measurement (pBSM).} [3 ebits]
\begin{equation}
  \text{pBSM} = \{\ket{\psi_\pm},\ket{00},\ket{11}\}. \label{eq:pBSM}
\end{equation}
This is the only solution that is not iso-entangled.
Two of the eigenstates are product states, while the two other are maximally entangled. It corresponds to a full-rank partial BSM of the kind achievable with linear optics without auxiliary photons\footnote{With a single ebit one can realize the BSM and coarse-grain two of the outcomes, leading to a measurement of only three outcomes.}. 

\item \emph{Twisted Bell state measurement (tBSM).} [3 ebits]
\begin{equation}
  \begin{aligned}
    \text{tBSM} &= \{(\ket{0+}\pm\ket{11})/\sqrt{2}, (\ket{0-}\pm\ket{10})/\sqrt{2}\}\\
    &= \text{CRot}_{\frac{\pi}{2}}(\text{BSM}).
  \end{aligned}
  \label{eq:tbsm}
\end{equation}

This is an iso-entangled basis with tangle$=\frac{1}{2}$.
Note that it can be obtained by applying a controlled $\pi/2$ rotation to the Bell state measurement, in analogy to the twisted product basis \eqref{eq:twisted}. 
It is locally equivalent to the tilted Bell state measurement introduced in Ref.~\cite{Boreiri2023} and part of the Bell family of iso-entangled two-qubit measurements defined in Ref.~\cite{DelSanto2024}.
     
\item \emph{The $B_2$ measurement.} [3 ebits]
\begin{equation}
  B_2 = \{(\ket{1-}\pm \ket{01})/\sqrt{2}, (\ket{1+}\pm \ii\ket{00})/\sqrt{2}\} \,. \label{eq:B2}
\end{equation}

The vectors of this measurement basis all have tangle $t=\frac{1}{2}$. It is a member of the Bell family of iso-entangled two-qubit measurements \cite{DelSanto2024}. This is the only measurement at the second level of our hierarchy that is not part of the second level of the Clifford hierarchy\footnote{In fact, it is not a member of the Clifford hierarchy at least until level six.}.
\end{itemize}

Every basis, except the partial Bell state measurement, is iso-entangled, i.e., the eigenstates associated to all the outcomes have the same degree of entanglement.
In Appendix~\ref{app:properties} we compute the auto-correlations between the Bloch vectors of the reduced states for $E_2$, tBSM, $B_2$ and EJM. These are invariant under the equivalence relations (Definition~\ref{def:equivalence}), which proves that all above bases, including those with the same tangle, are indeed distinct.

Since the results above are based on specific localization protocol, an interesting question is whether the entanglement consumption is minimal. 

In the modified Vaidman scheme developed in Ref.~\cite{Clark2010}, which localizes measurements with a finite average entanglement cost, the measurements that can be localized at a finite worst-case cost can be easily explored through exhaustive enumeration, as detailed in Appendix~\ref{app:clark}. With one ebit, the only option for the parties is to teleport the full state to one of them, similar to the Vaidman scheme. However, if we restrict operations to those that can be described by Pauli rotations, Alice and Bob cannot localize any measurements other than the product measurement.

At the next level, using three ebits, Alice and Bob can deterministically implement Pauli rotations with an angle $\theta=\pi/2$. With these rotations, they can localize the BSM by performing the rotation along, for example, the $X \otimes X$ axis. It is only at the subsequent level, requiring five ebits, that Alice and Bob can localize the $\pi/2$-twisted basis measurement \eqref{eq:twisted}. At this level, they can also localize the tBSM \eqref{eq:tbsm}, but not other measurements that could be localized with just three ebits in the Vaidman scheme.
These observations indicate that, despite requiring less entanglement on average, the scheme in Ref.~\cite{Clark2010} allows the localization of fewer measurements at low cost.

\section{Generalizations}
\label{sec:generalisations}

A complete classification of measurements requires several generalizations. As we will see, the methods of the last section straightforwardly generalize to (i) more ebits, (ii) higher dimensions, and even (iii) multiple parties. As explained briefly below and in detail in Appendix~\ref{app:vaidman}, the method is subject to combinatorial explosion, making it unpractical to obtain complete results with standard computational power. For this reason, we explore some of these generalizations numerically in this section.

\subsection{Higher levels of the hierarchy}

The sets of equations that determine the measurements localizable at higher levels of the Vaidman scheme can be solved recursively using the methods of Appendix~\ref{app:vaidman}. For example,
Equation~\eqref{eq:9ebit}, which characterizes the measurements that can be localized in the third level of the scheme, has exactly the same form as Eq.~\eqref{eq:app:3ebit-2}, so it can be solved via the same procedure.
This leads to a set of inertwinners $\{\mathcal{M}_{a_1,a_2,b}\}_{a_1,a_2}$, which enter the right-hand side of Eqs.~\eqref{eq:imp1} and \eqref{eq:imp2}, which in turn determine the $\mathcal{M}_b$s appearing in the one-ebit equation, Eq.~\eqref{eq:app:1ebit}, which can be solved to obtain the measurements $M$.
This logic can be extended to arbitrary levels of the scheme.
However, the complexity of the method grows exponentially with the number of ebits, and becomes intractable. 

We thus resort to a numerical search to explore higher levels, detailed in Appendix~\ref{app:heuristic}. Based on this search we make the following conjecture:
\begin{conjecture}[Localisable measurements with nine ebits]
  There are 47 equivalence classes of two-qubit measurements that can be performed locally with nine ebits. These include the eight classes that can be localized with three ebits. All states have tangles in the set $\left\{0,1,\frac{1}{2},\frac{1}{4},\frac{3}{4},\frac{1}{8},\frac{3}{8},\frac{5}{8},\frac{1}{4}(2\pm\sqrt{2}),\frac{1}{8}(2\pm\sqrt{2}),\frac{1}{8}(3\pm\sqrt{2})\right\}$.
\end{conjecture}

Among these solutions (a full list of which can be found in \cite{compapp}) we highlight 
the equivalents of the twisted product and Bell measurements, respectively Eqs.~\eqref{eq:twisted} and \eqref{eq:tbsm}, with a twist of $\pi/3$, and an iso-entangled basis with tangle $5/8$ in which all reduced Bloch vectors point to the vertices of a tetrahedron, a property reminiscent of the EJM in Eq.~\eqref{eq:EJM}.

Similarly, the methods of Appendix~\ref{app:clark} for the Clark \textit{et al.}  scheme are also subject to combinatorial explosion due to the possibility of using the ebits to implement either Pauli chains of small angles or multiple chains of larger angles. Despite this, the next level in the scheme after those analyzed in Section~\ref{sec:classification} only requires five ebits and can also be studied exhaustively. As previously noted, all measurements found at this level can be localized with three ebits in the Vaidman scheme.

\subsection{Higher dimensions}

Another important generalization is the extension to higher-dimensional systems. This approach provides a pathway towards a systematic understanding of high-dimensional joint measurements, an area that remains poorly understood. Indeed, our current knowledge is limited and extends little beyond generalizations of Bell state measurements.

The Vaidman scheme can be straightforwardly generalized to any dimension. The parties now share the maximally entangled qu$d$it states $\ket{\phi_+} = \sum_{j=0}^{d-1}\ket{j}\otimes\ket{j}$, where $\ket{j}$ is the computational basis. To teleport the qu$d$it systems back and forth, Alice and Bob measure the higher-dimensional Bell basis, $\{X^{x_1}_d Z^{x_2}_d \otimes \mathbb{1}_d \ket{\phi_+}\}_{x_1,x_2}$, where $X_d$ and $Z_d$ are the clock and shift matrices (the $d$-dimensional generalizations of the Pauli's), given by \begin{equation}
  X_d = \sum_{j=0}^{d-1}\ket{j+1}\bra{j},\quad Z_d = \sum_{j=0}^{d-1}\omega^j\ket{j}\bra{j},
\end{equation}
where $\omega=e^{2\pi \ii/d}$.

The equations that define the measurements that can be localized with a finite number of ebits in the Vaidman scheme (Eqs.~\eqref{eq:1ebit}-\eqref{eq:9ebit}) are thus straightforwardly generalized to higher dimensions and can be solved using the procedure in Appendix~\ref{app:vaidman}.

For example, qu$d$it-qu$d$it measurements $M$ that can be localized with one e$d$it are given by the solutions to the following generalization of \eqref{eq:1ebit}:
\begin{equation}
    M^\dagger \cdot\left(\mathbb{1}_d \otimes X^{b_1}_d Z^{b_2}_d\right)\cdot M = P_{b_1,b_2},
\end{equation}
where $b_1,\,b_2 = 0,1,\ldots,d-1$, and $P_{b_1,b_2}$ are $d \times d$ permutation matrices with phases associated to representations of the Weyl-Heisenberg group in dimension $d$.

We note that in our general method Appendix~\ref{app:vaidman}, we formulated the problem exploiting the fact that the Pauli group elements are also the generators of the $SU(2)$ Lie algebra. This allowed us to restrict to Hermitian permutation matrices with phases, which significantly reduces the problem. Aside from the fact that this property is lost in higher dimensions, the method is unchanged: one simply looks for all subsets of permutations with phases that satisfy the Heisenberg-Weyl group relations.
Unfortunately, this generalization also has high combinatorial cost.
Even for $d=3$ it becomes impossible to systematically find all of them, since there are $ 9! \choose 9$ possible subsets of permutations with phase that could form a representation of the Weyl-Heisenberg group in dimension 3.

Nevertheless, an interesting observation is that, for a given representation $\{P_b\}_b$ associated to a measurement, one can look for the natural generalization of this representation in higher dimensions, and subsequently compute the corresponding measurement as the intertwiner between the representation defined by $\{P_b\}_b$ and that given by the clock and shift matrices.
For example, for the Bell basis, one can write $P_X = \sigma_X \otimes \mathbb{1}$ and $P_Z = \sigma_Z \otimes \sigma_X$. It is straightforward to verify that $P_{X} = X_d \otimes \mathbb{1}_d$ and $P_{Z} = Z_d \otimes X_d$ satisfy the braiding relations, $P_Z P_X = \omega P_X P_Z$, and the normalisation conditions, $P_X^d = P_Z^d = \mathbb{1}_d$.
Thus, they form a representation of the Weyl-Heisenberg group for any dimension $d$.
The corresponding intertwiner, $M$, computed by following the steps in Appendix~\ref{app:vaidman}, is (up to local unitaries) the generalisation of the Bell basis to higher dimensions, $\{X^{x_1}_d Z^{x_2}_d \otimes \mathbb{1}_d \ket{\phi_+}\}_{x_1,x_2}$.

In the same way, one can look for generalizations of the other measurement classes. In Appendix~\ref{app:twistedbell} we explore this for the twisted Bell basis.

\subsection{Multiple parties}

Reference \cite{Vaidman2003} proposes a generalization of the bipartite scheme discussed in Section~\ref{sec:classification} to any number of parties. In this scheme, the state is teleported cyclically among the parties using multiple ebit channels to account for all possible distortions introduced by the teleportations. One of the parties attempts to correct these distortions and sends the state back, potentially performing a measurement on it if the previous teleportations were non-distorting. Similar to the bipartite case, a finite-resource version of this protocol can also be considered.

As an illustration, consider the case of three parties (Alice, Bob, and Cindy).
Bob and Cindy each use one an ebit to teleport the state to Alice, who tries to apply $M^\dagger$ to the full state and measures in the computation basis.
Thus, tripartite joint measurements that can be localized in this way (requiring two ebits) satisfy the equation
\begin{equation}
  M^\dagger \left(\mathbb{1}\otimes\sigma_b\otimes\sigma_c\right) M = P_{b,c}  \, .
  \label{eq:3party}
\end{equation}

We note that this equation can in principle also be solved using the methods presented in Appendix~\ref{app:vaidman}, but the computational cost is prohibitive.

As in the bipartite case, a subset of measurements that can be performed deterministically at the first level of the scheme includes Clifford operations, which  include the GHZ basis \cite{Clark2010}.
The heuristic search reveals a total of 13 bases. Removing the GHZ and product bases, and bases where two parties perform one of the measurements localizable with one ebit, we are left with seven new bases, all of them twisted. These are twisted bases in which, depending on the eigenstate of one of the parties, the other two measure in a different basis (which can be Bell-Bell in a different basis, product-product in a different basis, and product-Bell), and one basis where depending on the outcome of the first party the other two swap their measurement bases. Notably, there are no bases with tripartite entangled eigenstates other than the GHZ basis.
We list all the bases in the computational appendix \cite{compapp}.

More interestingly, in the second level of the hierarchy the finite scheme goes as follows:

\begin{enumerate}
  \item Bob and Cindy teleport their qubits to Alice with BSM outcomes \( b_0 \) and \( c_0 \).
  \item Alice applies \( M^\dagger \).
  \item Alice teleports the three qubits to Bob with outcomes \( a_1, a_2, a_3 \).
  \item Bob teleports the three qubits to Cindy using one of four ports, each utilizing three ebits, depending on his previous \( b_0 \), with BSM outcomes \( b_1, b_2, b_3 \).
  \item Cindy applies the following unitary to each of her four ports and measures in the computational basis: 
  \[
M^\dagger \cdot \left( \mathbb{1} \otimes \sigma_{b_0} \otimes \sigma_{c_0} \right) \cdot M \cdot \mathbf{\sigma_{A}} \cdot M^\dagger \cdot \left( \mathbb{1} \otimes \sigma_{b_0} \otimes \sigma_{c_0} \right) \cdot M
  \]
where \( \mathbf{\sigma_{A}} = \sigma_{b_1} \sigma_{a_1} \otimes \sigma_{b_2} \sigma_{a_2} \otimes \sigma_{b_3} \sigma_{a_3} \).
  \end{enumerate}

Hence, measurements $M$ that can be localized in this way satisfy the equations 
\begin{equation} \label{eq:2ndlevel3qubit}
  \mathcal{M}_{b_0c_0} \sigma_{b_1} \sigma_{a_1} \otimes \sigma_{b_2} \sigma_{a_2} \otimes \sigma_{b_3} \sigma_{a_3}  \mathcal{M}_{b_0c_0}^\dagger= P_{a_1,a_2,a_3,b_0,b_1,b_2,b_3,c_0} \,,
\end{equation}
where we have defined $\mathcal{M}_{b_0c_0} = M^\dagger\cdot \left( \mathbb{1} \otimes \sigma_{b_0} \otimes \sigma_{c_0} \right)\cdot M$. This is the direct generalization of Eq. \eqref{eq:9ebit} to the tripartite case and can in principle be solved analytically using the methods outlined previously. At this level, the scheme consumes 17 ebits. 

The heuristic search (see Appendix~\ref{app:heuristic}) reveals a total of 121 bases. We find two particularly interesting examples of tripartite qubit measurements that can be localised at this level of the scheme, which resemble the elegant joint measurement. For one of these solutions, the reduced single-party Bloch vectors of every party align with the vertices of a tetrahedron. For the other, all the vectors point to the vertices of a square. In Appendix \ref{app:3qubit} we give the explicit forms of these solutions. Both of them, and the remaining bases localizable at this level, are also listed in the computational appendix \cite{compapp}.

\section{Discussion}\label{sec:discussion}

Understanding joint measurements in quantum theory is both foundationally clarifying and essential for quantum information processing between multiple parties. In this work we studied the problem of determining which measurements can be localized with a given amount of entanglement. To this end, we introduced finite-resource versions of well-known teleportation-based localization protocols \cite{Vaidman2003,Clark2010} and developed general analytical methods to find all measurements which can be localized at each step of these schemes. This induces a classification on joint measurements based on the entanglement cost of realizing them locally. At low levels of this hierarchy, we recover several well-known measurements with special symmetries that appear frequently in quantum information theory. For example, the Bell state measurement can be localized with one ebit, and the elegant joint and twisted Bell state measurements can be localized with three ebits.

These basic results motivate a systematic generalization of our classification to higher levels of the scheme, higher dimensions and more parties. Our methods, although generally applicable, are computationally demanding, and new techniques will be needed. We explored some of these generalizations numerically. Preliminary results suggest exciting possibilities. In particular, the prospect of finding new joint measurements with special properties with potential applications for quantum protocols, the potential of lifting qubit measurements to general dimensions (e.g. the twisted Bell state measurement) and to multiple parties (e.g. generalizing the EJM to multiple parties). The latter has potential applications to network nonlocality, in particular generalizing the elegant distribution \cite{Gisin2019} to more parties.  We note in this context that any localization protocol for a joint measurement used in establishing correlations in a network scenario can be used to extend said network to a larger network, simply by promoting each teleportation step to a network edge. It would be interesting to explore the implications of this for network nonlocality.

An appealing feature of our methods is that they are rooted in representation theory, which makes them applicable to higher-dimensional systems. We saw that the same structure is present in the generalization to multipartite scenarios.
This suggests that a deeper group-theoretical understanding of the problem of localization may lead to novel insights as well as new methods to complete the classification. In particular, our classification is both conceptually and formally related to the Clifford hierarchy which was introduced in the context of teleportation-based quantum computation \cite{Gottesman1999}. The level of the hierarchy in which a unitary gate appears can be interpreted as the complexity of implementing the unitary in a fault-tolerant way \cite{Buhrman2013,Speelman2016,Chakraborty2015}, and we have argued that the same is true for localizing measurements. 

In this work we quantified entanglement by the number of ebits, i.e., we always assumed that the parties shared maximally entangled states. It is an interesting open question whether these are optimal for localizing measurements. Port-based teleportation, which is more efficient with partially entangled states when the number of ports is high \cite{Mozrzymas2018} might suggest a negative answer, although it is unclear whether this is also the case for non-probabilistic and ideal protocols. Indeed, in the context of non-local computation it is known that the finite-consumption version of the Vaidman scheme, which used maximally entangled states, is more efficient than the optimal port-based protocol for specific classes of unitaries, suggesting that port-based protocols may be suboptimal \cite{Chakraborty2015}. 

The finite-ebit versions of the Vaidman scheme and its modification by Clark \textit{et al.} give rise to systematic classifications of joint measurements based on entanglement cost. These schemes are universal in the sense that any measurement can be localized given a suitable amount of ebits, even if that cost may be infinite (\cite{Miller2024}, see also Appendix~\ref{app:clark}). However, these schemes may not always be optimal. Indeed, it is known that, e.g., the basis defined by applying $\text{CRot}_{\frac{\pi}{4}}$ to the computational basis, which we did not find in our analysis of the lower levels of the Vaidman scheme, can be localized with three ebits using an adaptive scheme \cite{Groisman2003,Gisin2024}. Finding a general method to determine localization cost that is independent of any particular protocol is a big open question.

\emph{Note added:} A previous version of this paper stated that the EJM and the $B_2$ measurement cannot be localized using two ebits. This claim was based on a semidefinite programming relaxation of Eq.~\eqref{eq:generaldefinition} that substituted the tensor product structure by invariance under partial transposition. This relaxation allows for classical communication between the parties, which is not permitted in our scenario. The seemingly nontrivial results were affected by an incorrect handling of complex numbers in the numerical solver used. After correcting the issues, we no longer have evidence either for or against two-ebit non-localizability of EJM and $B_2$. Resolving this question will require new methods; a promising direction is an algebraic-geometry approach as in Ref. \cite{Akibue2025}.

\begin{acknowledgements}
  We thank Victor Gitton for his insights on the methods for the one-ebit equation, Micha\l{} Studzi\'nsky, Marco T\'ulio Quintino, and Jessica Bavaresco for their valuable discussions on port-based teleportation, Sadra Boreiri,  Antoine Girardin and Tam\'as Kriv\'achy for discussions on network nonlocality, Mirjam Weilenmann, Armin Tavakoli, and Miguel Navascu\'es for their input on SDP criteria, and Eric Chitambar, Florian Speelman and Harry Buhrman for highlighting connections to other works.

We acknowledge support from NCCR-SwissMAP and and from FWF (Austrian Science
Fund) through an Erwin Schr\"odinger Fellowship (Project
J 4699).

\emph{Author contributions}: JP, APK and NG developed the theory, wrote the codes and performed calculations. JP and APK wrote the manuscript. All authors discussed the results and commented on the manuscript.

\end{acknowledgements}

\bibliography{refs}

\appendix
\onecolumngrid


\section{Clark \textit{et al.} scheme}\label{app:clark}
Clark \textit{et. al} introduced a localization scheme that, in contrast to the one discussed in the main text, consumes a finite amount of entanglement \emph{on average} \cite{Clark2010}.
Here, we present the basic idea of the scheme for qubits.
For details, as well as the generalization to arbitrary dimensions, we refer the reader to Ref.~\cite{Clark2010}.

\textbf{General scheme.}
The scheme is similar to the Vaidman scheme in the sense that the parties perform a series of back-and-forth teleportation rounds.
However, instead of trying to apply the unitary $M^\dagger$ (that, recall, rotates the desired measurement basis to the computational basis) directly, the parties decompose the unitary into a series of Pauli rotations, $M = \Pi_{\bm{k}} R_{\bm{k}}(\theta_k)$, where $R_{\bm{k}}(\theta) =\exp(\ii \theta \sigma_{k_1} \otimes \sigma_{k_2})$.
The vector $\bm{k} \in \{0,\dots,3\}^2$ specifies a Pauli string and $\theta_k$ is a rotation angle.
This decomposition into Pauli rotations offers a distinct advantage: Let $\sigma_{\bm{l}}=\sigma_{l_1}\otimes\sigma_{l_2}$ denote a teleportation-induced Pauli distortion.
When such a distortion happens, the state of the system is impacted in only one of two ways: either (i) the Pauli distortion $\sigma_{\bm{l}}$ commutes with the Pauli rotation $R_{\bm{j}}(\theta)$, i.e. $ R_{\bm{j}}(\theta)\sigma_{\bm{l}} = \sigma_{\bm{l}} R_{\bm{j}} (\theta)$, or (ii) the rotation and the distortion do not commute.
In this latter case, however, the relation $R_{\bm{j}}(\theta)\sigma_{\bm{l}} = \sigma_{\bm{l}} R_{\bm{j}}(-\theta)$ holds.
This contrasts with the Vaidman scheme, where the consequences of every Pauli distortion is typically distinct, and allows Alice and Bob to implement tailored correction mechanisms that allow both of them to reach a terminating condition.

The general mechanism proposed in Ref.~\cite{Clark2010} proceeds in a way similar to that of the Vaidman scheme, namely concatenating teleportations.
The main difference with the Vaidman scheme is that the parties know if their previous teleportation produced or not a distortion that commutes with the Pauli rotation they are trying to implement.
When the distortion and the rotation commute, the corresponding party keeps the state and begins the protocol of implementing the next Pauli rotation.
Otherwise, the party knows that the other party has performed $R_{\bm{j}}(-\theta)$ instead, and thus applies $R_{\bm{j}}(2\theta)$ in order to correct the state before teleporting it back.
The party that receives the state follows in the same fashion: if their previous teleportation produced a distortion that commutes with the Pauli rotation, they keep the state and apply the next Pauli rotation.
Otherwise, they know that the accumulated rotation in the state is $R_{\bm{j}}(-3\theta)$, so they apply $R_{\bm{j}}(4\theta)$ before teleporting back.
The scheme continues in this fashion until one of the parties is certain that the $R_{\bm{j}}(\theta)$ has been successfully applied, in which case the protocol continues for the next Pauli rotation.

Any unitary in a bipartite system can be expressed as a concatenation of Pauli rotations \cite{helgason1979differential}, which makes the scheme of Ref.~\cite{Clark2010} complete.
Moreover, the fact that both parties have a termination criterion allows for a finite average consumption of entanglement.
However, the parties do not know at which point the other party has successfully applied the rotation and passed onto the next Pauli rotation.
This makes that a new chain of ebits is needed for each possible exit point in order to implement the next rotation in the chain.

\textbf{Finite consumption scheme.}
The finite-consumption version of the protocol above was also developed in Ref.~\cite{Clark2010}, and applies when the Pauli rotation chains consist of binary angles, $\theta=(2m-1)\pi/2^D$ where $m$ is an integer.
For these angles, the correct rotation is also implemented when concatenating $D+1$ erroneous rotations, since $R(-\sum_{n=0}^{D+1}2^n\theta) = R(\theta)$.
Therefore, the localization of a binary-angle Pauli rotation is guaranteed to succeed, in the worst-case scenario, in $D+1$ ping-pong teleportation steps.
Since each teleportation steps consumes two ebits (except the first one, which consumes one), a Pauli rotation of angle $\theta=(2m-1)\pi/2^{D}$ can be implemented deterministically consuming at most $2D+1$ ebits.

\textbf{Localizable measurements with few ebits.}
Now we proceed to describe the measurements that can be localized deterministically following the Clark \textit{et. al} scheme for a finite number of ping-pong teleportation rounds.
As discussed above, the measurements localizable within this scheme with finite resources are those admitting a decomposition into Pauli rotations with binary angles, $\theta=(2m-1)\pi/2^D$ for $m$ and $D$ integers.
The deterministic implementation of each of such rotations consumes at most $2D+1$ ebits, although one must bear in mind that when concatenating rotations one needs that amount of ebits per possible exit point of the previous rotation.

The lower steps in the scheme can be explored easily by exhaustive enumeration, since the measurements $M$ that can be localized have the form
\begin{equation}
  M=\prod_{r=1}^N\exp\left(-\frac{\ii}{2}(2m-1)\frac{\pi}{2^{D_{\bm{k}^{(r)}}}}\sigma_{k^{(r)}_1}\otimes\sigma_{k^{(r)}_2}\right)
\end{equation}
for $\bm{k}^{(r)}=(k^{(r)}_1,k^{(r)}_2)\in\{0,1,2,3\}^2$ and $m$, $D_{\bm{k}^{(r)}}$ integers.
The cost in ebits of localizing $M$ is $\sum_{r}\prod_{s<r}(2D_{\bm{k}^{(s)}}+1)$.

When only one one ebit is available, only one rotation (of angle $\theta=\pi$) is possible.
For all Pauli strings given by $\bm{k}$ the corresponding rotation is a permutation matrix.
This means that the only thing that the parties can do is to teleport the full state to one of them, in analogy with the Vaidman scheme.
However, if in contrast to the Vaidman scheme we restrict to performing only operations that can be described by Pauli rotations, Alice and Bob can not localize any measurements other than the product measurement.

The next level, which uses three ebits, allows Alice and Bob to deterministically implement a single Pauli rotation of angle $\theta=\pi/2$.
With these, they can localize the BSM by performing the rotation along, for instance, the $\sigma_X\otimes \sigma_X$ axis.

It is only at the next level, that consumes five ebits, when Alice and Bob can localize the $\pi/2$ twisted basis measurement \eqref{eq:twisted}.
This is because the twisted basis measurement can be written as a single $\pi/4$ rotation along the $\sigma_X\otimes \sigma_Z$, $\sigma_Y\otimes \sigma_Z$, $\sigma_Z\otimes \sigma_X$ or $\sigma_Z\otimes \sigma_Y$ axes.
At this level they can also localize the tBSM \eqref{eq:tbsm}, but not the other measurements that could be localized with just three ebits in the Vaidman scheme.
This means that, despite being asymptotically more favourable, the scheme of Ref.~\cite{Clark2010} allows to localize fewer measurements than that in Ref.~\cite{Vaidman2003} for finite and small shared entanglement.
This may be remedied if we allow the final party to perform an arbitrary operation after the last correction of the final Pauli rotation.
This would, following the spirit of the operation that Alice performs in the first level of the Vaidman scheme.

\textbf{Upper bounds to the entanglement cost of localization.}
While we have seen above that the finite-ebit version of the Clark \textit{et. al} protocol of Ref.~\cite{Clark2010} seems less powerful than the finite-ebit version of the Vaidman protocol of Ref.~\cite{Vaidman2003} for small number of ebits, an appealing property of the Clark \textit{et. al} protocol is that it allows to easily provide upper bounds to the maximum amount of ebits necessary to localize a particular measurement.
The reason for this is that any two-qubit unitary can be decomposed into a sequence of only seven Pauli rotations.
This is known as the Cartan decomposition of $SU(4)$ \cite{helgason1979differential}, which has found many applications in quantum computing \cite{khaneja2000cartan,Kraus2001,Vatan2004,Vidal2004}.
The decomposition reads

\begin{equation}
  M = V_A \otimes V_B e^{\frac{\ii}{2} \xi_1 \sigma_X \otimes \sigma_X} e^{\frac{\ii}{2} \xi_2 \sigma_Y \otimes \sigma_Y} e^{\frac{\ii}{2} \xi_3 \sigma_Z \otimes \sigma_Z} W_A \otimes W_B,
  \label{eq:cartan}
\end{equation}
where  $V_A$, $V_B$, $W_A$ and $W_B$ are local qubit unitaries, so they can be written as sequences of Pauli rotations in terms of their Euler angles, $R_Z(\alpha)R_Y(\beta)R_Z(\gamma)$.

While this decomposition may not be optimal in terms of the number of ebits consumed to localize a given measurement, it provides an upper bound to the entanglement cost of localization.
In order to localize an $M$ given in the decomposition of Eq.~\eqref{eq:cartan}, the parties first apply locally $W_A$ and $W_B$ to their respective shares of $\ket{\psi}$, and then perform the ping-pong protocol implementing, first, the Pauli rotations for the nonlocal angles, and then the Pauli rotations corresponding to the Euler decompositions of $V_A$ and $V_B$.
Note that the final rotations in the Euler decompositions are along the corresponding $Z$ axes.
These rotations do not modify the statistics of measurements in the computational basis, and thus can be ignored.
Thus, the fact that any two-qubit measurement can be expressed as
\begin{equation}
  \begin{aligned}
    M =& \left[R_Z\left(k_1\frac{\pi}{2^{D_1}}\right)\otimes R_Z\left(k_2\frac{\pi}{2^{D_2}}\right)\right]\cdot\left[R_X\left(k_3\frac{\pi}{2^{D_3}}\right)\otimes R_X\left(k_4\frac{\pi}{2^{D_4}}\right)\right] \cdot \left[R_Z\left(k_5\frac{\pi}{2^{D_5}}\right)\otimes R_Z\left(k_6\frac{\pi}{2^{D_6}}\right)\right] \\
    & \cdot R_{XX}\left(k_7\frac{\pi}{2^{D_7}}\right) \cdot R_{YY}\left(k_8\frac{\pi}{2^{D_8}}\right) \cdot R_{ZZ}\left(k_9\frac{\pi}{2^{D_9}}\right) \cdot \left(W_A\otimes W_B\right),
  \end{aligned}
  \label{eq:clark}
\end{equation}
(with $k_i,\,i=1,\ldots,9$ being odd integers, and $D_i$ being integers or $\infty$) means that $M$ can be localized deterministically using, at most, $\sum_{i=3}^9\prod_{D_k\not=\{0,\infty\},k \geq i}(2D_i+1)$ ebits.

When applied to the measurements discussed in the main text, one finds that a maximum number of ebits needed of 462 for the pBSM of Eq.~\eqref{eq:pBSM}, 278 for the EJM of Eq.~\eqref{eq:EJM}, 203 for the $E_2$ basis of Eq.~\eqref{eq:E2} (using as Euler angles $ZYZ$ instead of $ZXZ$), and 1022 for the tBSM basis of Eq.~\eqref{eq:tbsm}.
For the $B_2$ measurement of Eq.~\eqref{eq:B2}, the decomposition involves angles that are not odd multiples of $\pi/2^D$, and thus the upper bound is infinite.

\section{Recursive method for finding all measurements that can be localized with the finite-ebit Vaidman scheme}\label{app:vaidman}

In this appendix we describe a general method for finding all measurements $M$ that can be localized through the finite-ebit Vaidman scheme.
We focus on the cases of one and three ebits.
This is, we will find all solutions to the equations \eqref{eq:1ebit} and \eqref{eq:3ebit}.
The generalisation to higher steps in the scheme, i.e. solving Eq.~\eqref{eq:9ebit} and beyond, is straightforward but computationally more demanding.

First, it is important to note that there are infinitely many solutions to any of the relevant equations. We are only interested in finding a representative of each equivalence class, Definition~\ref{def:equivalence}. 

\subsection{The one-ebit equation \eqref{eq:1ebit}}
Let us recall the first equation to solve, which characterizes all measurements that can be localized in the first level of the finite-consumption Vaidman scheme, which uses one ebit.
Any such measurement has a least one representation in the equivalence class (Definition~\ref{def:equivalence}), which satisfies
\begin{equation} \label{eq:app:1ebit}
    M^\dagger\cdot\mathbb{1}\otimes\sigma_b\cdot M = P_b ,
\end{equation}
for all $b$ and where $P_b=\tilde{P}_b \cdot \Phi_b$ are permutation matrices with phases.
This equation states that $M$ is localizable if and only if it transforms between two representations of the Lie algebra $SU(2)$, that given by the identity tensorized with the Pauli matrices, and that given by $\{P_b\}_b$.
In representation theory, the unitary map $M$ is called an \emph{intertwiner}.
Hence, all measurements $M$ that can be localized with one ebit are all intertwiners between the representation $\mathbb{1}\otimes\sigma_b$ and representations of the form $\tilde{P}_b \cdot \Phi_b$ of $SU(2)$.
We devide the problem of finding these $M$  in two steps:
\begin{enumerate}
    \item Find all representations of the Pauli group of the form $\{P_b\}_b=\{\tilde{P}_b \cdot \Phi_b\}_b$.
    \item Given a representation, defined by a triplet $\{P_b\}_b$, find all intertwiners between $\mathbb{1}\otimes\sigma_b$ and $\tilde{P}_b \cdot \Phi_b$.
    As we will show below, given a fixed representation $\{P_b\}$, the intertwiner is unique within the equivalence class defined above.
\end{enumerate}

A priori, for the case of two qubits, there are $4!=24$ possible permutation matrices $\tilde{P}$, supplemented by arbitary phases.
However, since $\mathbb{1}\otimes\sigma_b$ is also a representation of the Lie algebra $SU(2)$, every element of the representation $P_b$ must be Hermitian, and have the same spectrum as $\mathbb{1}\otimes\sigma_b$.
This reduces to number of possible permutations to the ten symmetric ones and imposes constraints on the possible phases, leading to the following 21 possible $P_b$:

\begin{align*}
  P_1 &=  \left(
\begin{array}{cccc}
0 & 0 & 0 & e^{\ii a_1} \\
0 & 0 & e^{\ii b_1} & 0 \\
0 & e^{-i b_1} & 0 & 0 \\
e^{-\ii a_1} & 0 & 0 & 0 \\
\end{array} \,,
\right)
&P_2 = \left(
 \begin{array}{cccc}
  0 & 0 & 0 & e^{\ii a_2} \\
  0 & \mp 1 & 0 & 0 \\
  0 & 0 & \pm 1 & 0 \\
  e^{-\ii a_2} & 0 & 0 & 0 \\
 \end{array}
 \right) \,, \\
 P_3 &= \left(
   \begin{array}{cccc}
    0 & 0 & e^{\ii a_3} & 0 \\
    0 & 0 & 0 & e^{\ii b_3} \\
    e^{-\ii a_3} & 0 & 0 & 0 \\
    0 & e^{-\ii b_3} & 0 & 0 \\
   \end{array}
   \right) \,,
   &P_4 =\left(
     \begin{array}{cccc}
      0 & 0 & e^{\ii a_4} & 0 \\
      0 & \mp1 & 0 & 0 \\
      e^{-\ii a_4} & 0 & 0 & 0 \\
      0 & 0 & 0 & \pm1 \\
     \end{array}
     \right)\,,\\
     P_5 &= \left(
   \begin{array}{cccc}
    0 & e^{-\ii a_5} & 0 & 0 \\
    e^{\ii a_5} & 0 & 0 & 0 \\
    0 & 0 & 0 & e^{-\ii b_5} \\
    0 & 0 & e^{\ii b_5} & 0 \\
   \end{array}
   \right) \,, 
    &P_6 = 
   \left(
\begin{array}{cccc}
0 & e^{-\ii a_6} & 0 & 0 \\
e^{\ii a_6} & 0 & 0 & 0 \\
0 & 0 & \mp1 & 0 \\
0 & 0 & 0 & \pm 1 \\
\end{array}
\right) \,, \\
P_7 & = \left(
 \begin{array}{cccc}
  \mp 1 & 0 & 0 & 0 \\
  0 & 0 & 0 & e^{\ii a_7} \\
  0 & 0 & \pm 1 & 0 \\
  0 & e^{-\ii a_7} & 0 & 0 \\
 \end{array}
 \right) \,,
 &P_8 = \left(
   \begin{array}{cccc}
    \mp 1 & 0 & 0 & 0 \\
    0 & 0 & e^{\ii a_8} & 0 \\
    0 & e^{-\ii a_8} & 0 & 0 \\
    0 & 0 & 0 & \pm1 \\
   \end{array}
   \right) \,,\\
   P_9 &= \left(
     \begin{array}{cccc}
      \mp 1 & 0 & 0 & 0 \\
      0 & \pm1 & 0 & 0 \\
      0 & 0 & 0 & e^{\ii a_9} \\
      0 & 0 & e^{-\ii a_9} & 0 \\
     \end{array}
     \right) \,,
     &P_{10} = \text{diag}(\text{perm}(-1,-1,1,1)) \,,
\end{align*}
where $a_i,b_i \in [0,2\pi)$ are arbitrary phases.

Next, among the set of $P$ identified above, one needs to find \emph{all} subsets of three $P_b$ that form a representation of the Pauli group.
To this end, it is sufficient to identify triplets $\{P_b\}_b$ that satisfy the $SU(2)$ algebra:
\begin{equation}
  [P_i,P_j] = 2i\epsilon_{ijk}P_k \, ,
\end{equation}
where $\epsilon_{ijk}$ is the Levi-Civita symbol. 

Finding these subsets is the main bottleneck in the method.
A priori, we have $21!/(21-3)!=7980$ possible ordered subsets to check.
However, it follows from the following Lemma that different orderings of $P_b$ lead to the same equivalence class.

\begin{lemma}
    If $M$ is a solution to Eq.~\eqref{eq:app:1ebit} for some $\{P_b\}_b$, then for any permutation $\pi$, the measurement $M \rightarrow \left[\mathbb{1} \otimes \left(e^{\ii \sigma_X k_1 \frac{\pi}{4}} \cdot e^{\ii \sigma_Z k_2\frac{\pi}{4}}\right)\right] \cdot M$ (corresponding to a local basis change on Bob), is a solution to Eq.~\eqref{eq:app:1ebit} for the permuted $P_{\pi(b)}$, for some $k_1,k_2 = 0,1$.
    \label{lemma:unique}
  \end{lemma}

  \begin{proof}
    Assume $M$ satisfies \eqref{eq:app:1ebit} for some $\{P_X,P_Y,P_Z\}$.
    For clarity, assume we permute $P_Y \leftrightarrow P_Z$ (the extension to other permutations is straightforward). 
   
    Define the matrix $M' = \mathbb{1} \otimes e^{\ii\sigma_Z \frac{\pi}{4}}\cdot  M$.
    This matrix satisfies Eq.~\eqref{eq:app:1ebit} for the permuted $\{P_X,P_Z,P_Y\}$.    
    Indeed, plugging $M'$ into Eq.~\eqref{eq:app:1ebit} gives:

    \begin{align}
        M'^\dagger \cdot \mathbb{1} \otimes \sigma_X \cdot M' &= M^\dagger \cdot \mathbb{1} \otimes e^{-\ii \sigma_Z \frac{\pi}{4}} \cdot \sigma_X \cdot e^{\ii \sigma_Z \frac{\pi}{4}} \cdot M \\
        &= M^\dagger \cdot \mathbb{1} \otimes \sigma_Y \cdot M =  P_Y \,.
    \end{align}

    Similarly, $M'^\dagger \cdot \mathbb{1} \otimes \sigma_Y \cdot M' = P_Z$ and $M'^\dagger \cdot \mathbb{1} \otimes \sigma_Z \cdot M' = P_X$.
  \end{proof}

The above Lemma allows us to restrict our search to $21!/(21-3)!/3! = 1330$ possible subsets to check.
Importantly, some subsets of $P_i$ only form valid representations when certain relations between the phases are satisfied.

Once all valid representations are found, one can proceed to step 2).
Given a valid representation $\{P_b\}_b$, we can now find all intertwiners $M$ between $\mathbb{1}\otimes\sigma_b$ and $\tilde{P}_b \cdot \Phi_b$.
Before giving the procedure to construct $M$, we prove the following lemma, which shows that the intertwiner is unique up to the equivalence relations \ref{def:equivalence}.

\begin{lemma}[Uniqueness of solutions to Eq.~\eqref{eq:app:1ebit}]
Given a fixed valid representation $\{P_b\}_b$, the intertwiner $M$ between this representation and $(\mathbb{1}\otimes\text{Pauli})$ is unique up to local unitaries, i.e., if $M$ and $N$ are two solutions to Eq.~\eqref{eq:1ebit} for given $\{P_b\}$, then $N = (U \otimes \mathbb{1}) \cdot M $ for some local unitary $U$.
\end{lemma}

\begin{proof}
 Suppose that $M$ and $N$ satisfy Eq.~\eqref{eq:app:1ebit}.
 Then, we have that
 \begin{equation}
      M^\dagger \cdot \mathbb{1} \otimes \sigma_b \cdot M = N^\dagger \cdot \mathbb{1} \otimes \sigma_b \cdot N \,.
 \end{equation}
Multiplying both sides on the left by $N$ and on the right by $M^\dagger$ gives:
\begin{equation}
    V\cdot (\mathbb{1} \otimes \sigma_b)  = (\mathbb{1} \otimes \sigma_b ) \cdot V \,,
\end{equation}
where we have defined $V \equiv N \cdot M^\dagger $. It is easy to see that since $V$ commutes with $\mathbb{1} \otimes \sigma_b$ for all $b$, it must be of the form $V = U \otimes \mathbb{1}$ for some local unitary $U$. Hence, we have $N \cdot M^\dagger = U \otimes \mathbb{1}$, which implies $N = \left(U \otimes \mathbb{1}\right) \cdot M$.
\end{proof}

To explicitly construct $M$, we note that Eq.~\eqref{eq:app:1ebit} is a homogeneous Sylvester equation \cite{Sylvester1884,Bhatia1997}, which can be solved as a linear system.
Indeed, any matrix equation of the type $AX-XB=0$ can be rewritten as the linear system
\begin{equation}
  \left(\mathbb{1} \otimes A - B^T \otimes \mathbb{1}\right) \text{vec}(X) = 0 ,
  \label{eq:sylvester}
\end{equation}
where $\text{vec}(X)$ is the vectorization of the matrix $X$ obtained by stacking the columns of $X$ on top of each other.
In order to obtain Eq.~\eqref{eq:sylvester} one needs to use the identity $\text{vec}(RS) = (\mathbb{1} \otimes R)\text{vec}(S) =(S^T \otimes \mathbb{1})\text{vec}(R)$. The solutions $X$ to the Sylvester equation are thus in the null space of the matrix $\mathbb{1} \otimes A - B^T \otimes \mathbb{1}$.

Equation \eqref{eq:app:1ebit} is, recall, a set of four equations, one for each $b$.
The one corresponding to $b=0$ trivially imposes unitarity, and one of the remaining three can be obtained by multiplying the other two.
Localizable $M$s satisfy both remaining equations, which means that the null space of
\begin{equation} \label{eq:1ebitnullspace}
  \begin{pmatrix}
    \mathbb{1}_4 \otimes \mathbb{1}_2 \otimes \sigma_X - P^T_X \otimes \mathbb{1}_4 \\
    \mathbb{1}_4 \otimes \mathbb{1}_2 \otimes \sigma_Z - P^T_Z \otimes \mathbb{1}_4  
  \end{pmatrix}
\end{equation}
provides a basis $\{M_\alpha\}_\alpha$ for the solutions to Eq.~\eqref{eq:app:1ebit}.
The final solution $M$ will be of the form $M = \sum_\alpha \lambda_\alpha M_\alpha$, where $\lambda_\alpha$ are arbitrary complex coefficients.
These coefficients are fixed by demanding that $M$ is unitary. 

\subsection{The three-ebit equation \eqref{eq:3ebit}}
In the following we proceed to solve the blind teleportation equation for the second level of the Vaidman protocol, which consumes three ebits.
For completeness, let us reproduce Eq.~\eqref{eq:3ebit}:
\begin{equation}
  M^\dagger \cdot\left(\mathbb{1} \otimes \sigma_b\right)\cdot M \cdot\left(\sigma_{a_1} \otimes \sigma_{a_2}\right)\cdot M^\dagger \cdot\left(\mathbb{1} \otimes \sigma_b\right)\cdot M = P_{a_1,a_2,b},
  \label{eq:app:3ebit-1}
\end{equation}
which is a set of $4^3$ equations (one for each of $a_1,a_2,b\in\{0,X,Y,Z\}$).
If one writes $\mathcal{M}_b=M^\dagger \cdot\left(\mathbb{1} \otimes \sigma_b\right)\cdot M$, Eq.~\eqref{eq:app:3ebit-1} reads
\begin{equation}
  \mathcal{M}_b^\dagger \cdot\left(\sigma_{a_1} \otimes \sigma_{a_2}\right)\cdot \mathcal{M}_b = P_{a_1,a_2,b}.
  \label{eq:app:3ebit-2}
\end{equation}
It is easy to notice the similarity between Eq.~\eqref{eq:app:3ebit-2} and Eq.~\eqref{eq:app:1ebit}.
Indeed, we will see that the solution strategy is very similar, and that it can be generalised straightforwardly to higher levels.

In analogy to the first level of the protocol, the solutions $M$ to Eq.~\eqref{eq:app:3ebit-1} are intertwiners between the representation $\{\mathbb{1} \otimes \sigma_b\}_b$ and $\{\mathcal{M}_b\}_b$ of $SU(2)$.
The novelty resides in that this last representation, moreover, has to be composed of intertwiners between the representations $\{\sigma_{a_1} \otimes \sigma_{a_2}\}_{a_1,a_2}$ and $\{P_{a_1,a_2,b}\}_{a_1,a_2}$ of $SU(2)\otimes SU(2)$.
Therefore, the procedure to follow will consist of
\begin{enumerate}
  \item finding all representations $\{P_{a_1,a_2,b}\}_{a_1,a_2}$ of $SU(2)\otimes SU(2)$,
  \item computing the corresponding intertwinners $\mathcal{M}_b$,
  \item finding, among all $\mathcal{M}_b$s, the triples that form a representation of $SU(2)$, and
  \item finding the interwinners between these representations and that given by $\{\mathbb{1} \otimes \sigma_b\}_b$.
\end{enumerate}

The first step could in principle be solved by exhaustive search for sets of $4^2-1$ permutations $\{P_{a_1,a_2,b}\}_{a_1,a_2}$ that satisfy the $SU(2)\otimes SU(2)$ algebra. However, the search can be actually reduced to that of just permutations that satisfy the $SU(2)$ algebra, as in the previous section.
To see this, let us consider the particular cases $a_1=0$ and $a_2=0$.
For these cases, Eq.~\eqref{eq:app:3ebit-2} reads
\begin{align}
  \mathcal{M}_b^\dagger \cdot \left(\sigma_{a_1} \otimes \mathbb{1}\right)\cdot \mathcal{M}_b &= P_{a_1,0,b} , \label{eq:imp1} \\
  \mathcal{M}_b^\dagger \cdot\left(\mathbb{1} \otimes \sigma_{a_2}\right)\cdot \mathcal{M}_b &= P_{0,a_2,b} . \label{eq:imp2}
\end{align}
Moreover, if we multiply the two equations above, using the fact that $\mathcal{M}_b^\dagger\mathcal{M}_b=\mathbb{1}$ we obtain the left-hand side of Eq.~\eqref{eq:app:3ebit-2} on the left-hand side.
For the right-hand side, depending on the order of the multiplication, we obtain either $P_{a_1,0,b}P_{0,a_2,b}$ or $P_{0,a_2,b}P_{a_1,0,b}$.
Since the left-hand side coincides in both cases with that of Eq.~\eqref{eq:app:3ebit-2}, we have that $P_{a_1,0,b}P_{0,a_2,b}=P_{0,a_2,b}P_{a_1,0,b}=P_{a_1,a_2,b}$.
This implies, in particular, that $[P_{a_1,0,b}P_{0,a_2,b}]=0\,\forall\,a_1,a_2,b$.

The observations above allow to reduce the search of sets of $4^2-1$ permutations $\{P_{a_1,a_2,b}\}_{a_1,a_2,b}$ that satisfy the $SU(2)\otimes SU(2)$ algebra, to the search of \textit{pairs} of sets of $4-1$ permutations, $(\{P_{a_1,0,b}\}_{a_1},\{P_{0,a_2,b}\}_{a_2})$, that each satisfy the $SU(2)$ algebra, and that commute.
Importantly, finding sets of permutations that satisfy the $SU(2)$ algebra is exactly the same problem as the one we solved in the previous section.
This means that, among all the triplets found in the previous section, the only thing that remains is to find for which pairs the corresponding permutations commute.
This may impose additional constraints on the phases of the permutations.

Once a pair of sets of permutations, $(\{P_{a_1,0,b}\}_{a_1},\{P_{0,a_2,b}\}_{a_2})$, has been identified, one can solve the Sylvester equation associated to Eqs.~\eqref{eq:imp1}-\eqref{eq:imp2}.
The corresponding intertwiner, $\mathcal{M}_b$, is given by the null space of the matrix
\begin{equation} \label{eq:3ebitnullspace}
  \begin{pmatrix}
    \mathbb{1}_4 \otimes \mathbb{1}_2 \otimes \sigma_X - P^T_{0,X} \otimes \mathbb{1}_4 \\
    \mathbb{1}_4 \otimes \mathbb{1}_2 \otimes \sigma_Z - P^T_{0,Z} \otimes \mathbb{1}_4 \\
    \mathbb{1}_4 \otimes \sigma_X \otimes \mathbb{1}_2  - P^T_{X,0} \otimes \mathbb{1}_4 \\
    \mathbb{1}_4   \otimes \sigma_Z  \otimes \mathbb{1}_2  - P^T_{Z,0} \otimes \mathbb{1}_4 \\
  \end{pmatrix} ,
\end{equation}
where we have omitted the subindex $b$ for clarity.

Following the reasoning in Lemma~\ref{lemma:unique} it is easy to see that the intertwiner $\mathcal{M}_b$ is unique up to an arbitrary phase.
This phase will be important in the next steps.
Moreover, the uniqueness of the solution implies automatically that the matrices $\mathcal{M}_b$ will be unitary.

Thus, for each pair of sets of permutations, we find the corresponding intertwiner $\mathcal{M}_b$.
At the end of this search, we have a set of intertwiners $\{\mathcal{M}_b\}_b$ that satisfy Eq.~\eqref{eq:app:3ebit-2}.
In order to find the localizable measurements, $M$, we recall the definition of $\mathcal{M}_b$:
\begin{equation}
  \mathcal{M}_b = M^\dagger \cdot\left(\mathbb{1} \otimes \sigma_b\right)\cdot M.
\end{equation}
This is the one-ebit equation, Eq.~\eqref{eq:1ebit}-\eqref{eq:app:1ebit}, where the right-hand side is given by matrices $\mathcal{M}_b$ instead of permutations with phases.
Therefore, the procedure to find the localizable measurements $M$ is the same as the one described in the previous section, but with the right-hand side given by the $\mathcal{M}_b$.
Namely, we identify triples of matrices $\mathcal{M}_b$ that satisfy the $SU(2)$ algebra, and then find the corresponding intertwiners $M$ by solving the corresponding null space \eqref{eq:1ebitnullspace} with the $P$s substituted by $\mathcal{M}$s.

\section{Properties of localizable measurements} \label{app:properties}
Here, we give additional details on the properties of the  solutions to the third level of the finite-consumption Vaidman scheme, listed in Section~\ref{sec:classification}. We compute the autocorrelations $mm^A_{jk}=\vec m_A(j)\vec m_A(k)$ between the Bloch vectors of Alice and Bob (denoted $\vec m_A(j)$ and $\vec m_B(j)$, $j=1,2,3,4$,  respectively) of the reduced states for tBSM,  $B_2$, EJM and $E2$. These are different, which proves that the measurements with eigenvectors with the same tangle $t$ are not equivalent.
\begin{itemize}
\item Elegant Joint Measurement:
    \begin{equation}
    M_{\text{EJM}}=\frac{1}{\sqrt{8}}\begin{pmatrix}
      1+\ii & -1+\ii & 1-\ii & -1-\ii  \cr
      -2\ii & 0 & 0 & -2\ii  \cr
      0 & 2\ii & 2\ii &  0 \cr
      1-\ii & -1-\ii & 1+\ii & -1+i\end{pmatrix}
    \end{equation}
    This is an iso-entangled basis with $t=\frac{1}{4}$. It belongs to the Elegant family \cite{DelSanto2024}. The autocorrelations are given by
    \begin{equation}
      mm_A = mm_B=\frac{1}{4}\begin{pmatrix}
      3 & -1 & -1 & -1  \cr
      -1 & 3 & -1 & -1  \cr
      -1 & -1 & 3 & -1 \cr
      -1 & -1 & -1 & 3\end{pmatrix}\,.
    \end{equation}

    \item $E_2$ basis:
    \begin{equation}
    M_{E_2}=\frac{1}{2}\begin{pmatrix}
      0 & \sqrt{2} & -\sqrt{2} & 0  \cr
      1 & -1 & -1 & 1  \cr
      1 & 1 & 1 & 1 \cr
      \sqrt{2} & 0 & 0 & -\sqrt{2}\end{pmatrix} \, .
    \end{equation}
    This is the second iso-entangled basis with tangle$=\frac{1}{4}$.
    It can be written in terms of state-vectors as: 
    $M_{E_2}=\{(\psi^-\pm\ket{00})/\sqrt{2}, (\psi^+\pm\ket{11})/\sqrt{2}\}$. It belongs to the Elegant family defined in Ref.~\cite{DelSanto2024}.
    
    The Bloch vector correlation reads:
    \begin{eqnarray}
    mm_A=\frac{1}{4}\begin{pmatrix}
      3 & 1 & -3 & -1  \cr
      1 & 3 & -1 & -3  \cr
      -3 & -1 & 3 & 1 \cr
      1 & -3 & 1 & 3\end{pmatrix}\,,  \quad & mm_B=\frac{1}{4}\begin{pmatrix}
        3 & -3 & 1 & -1  \cr
        -3 & 3 & -1 & 1  \cr
        1 & -1 & 3 & -3 \cr
        -1 & 1 & -3 & 3\end{pmatrix} \,.
    \end{eqnarray}
    Note that these Bloch vector correlations are invariant under local unitaries. Hence, since they differ between EJM and Y, these two solutions are not locality equivalent. Alice and Bob's Bloch vectors read:
  
    \begin{align}
    \vec m_A(1) &= \left(\frac{1}{\sqrt{2}}, 0, -\frac{1}{2}\right) \,,\quad \quad &\vec m_B(1) = \left(\frac{1}{\sqrt{2}}, 0, -\frac{1}{2}\right)\,, \\
    \vec m_A(2) &= \left(\frac{1}{\sqrt{2}}, 0, \frac{1}{2}\right)\,, \quad \quad &\vec m_B(2) = \left(-\frac{1}{\sqrt{2}}, 0, \frac{1}{2}\right)\,, \\
    \vec m_A(3) &= \left(-\frac{1}{\sqrt{2}}, 0, \frac{1}{2}\right)\,, \quad \quad &\vec m_B(3) = \left(\frac{1}{\sqrt{2}}, 0, \frac{1}{2}\right)\,, \\
    \vec m_A(4) &= \left(-\frac{1}{\sqrt{2}}, 0, -\frac{1}{2}\right) \,,\quad \quad &\vec m_B(4) = \left(-\frac{1}{\sqrt{2}}, 0, -\frac{1}{2}\right)\,.
    \end{align}
  \item Twisted Bell basis
  \begin{equation}
  M_{tBSM}=\frac{1}{2}\begin{pmatrix}
    1 & 1 & 1 & 1  \cr
    1 & -1 & -1 & 1  \cr
    0 & \sqrt{2} & -\sqrt{2} & 0 \cr
    -\sqrt{2} & 0 & 0 & \sqrt{2}\end{pmatrix} \,.
  \end{equation}
  This is an iso-entangled basis with tangle$=\frac{1}{2}$. It belongs to the Bell family.
  It can be written in terms of state-vectors as: 
  $M_{\rm tBSM}=\{(\ket{0+}\pm\ket{11})/\sqrt{2}, (\ket{0-}\pm\ket{10})/\sqrt{2}\}$. The autocorrelations are:
  \begin{eqnarray}
  mm_A=\frac{1}{2}\begin{pmatrix}
    1 & -1 & 1 & -1  \cr
    -1 & 1 & -1 & 1  \cr
    1 & -1 & 1 & -1 \cr
    -1 & 1 & -1 & 1\end{pmatrix} \,, \quad  &mm_B=\frac{1}{2}\begin{pmatrix}
      1 & -1 & -1 & 1  \cr
      -1 & 1 & 1 & -1  \cr
      -1 & 1 & 1 & -1 \cr
      1 & -1 & -1 & 1\end{pmatrix} \,,
  \end{eqnarray}
  and the Bloch vectors read:
  \begin{align}
  \vec m_A(2)&=\left(\frac{1}{\sqrt{2}}, 0, 0\right), &  \vec m_B(2)=\left(-\half, 0, \half\right),  \\ 
  \vec m_A(1)&=\left(-\frac{1}{\sqrt{2}}, 0, 0\right), &  \vec m_B(1)=\left(\half, 0, -\half \right), \\ 
  \vec m_A(3)&=\left(-\frac{1}{\sqrt{2}}, 0, 0\right), &  \vec m_B(3)=\left(-\half, 0, \half\right),  \\ 
  \vec m_A(4)&=\left(\half, 0, 0\right), &  \vec m_B(4)=\left(\half, 0, -\half\right).
  \end{align}
  The tBSM basis is locally unitarily equivalent to the basis:
  \begin{equation}
    tBSM \sim \frac{1}{2}\begin{pmatrix}
      c & s & 0 & 0  \cr
      0 & 0 & c & s  \cr
      0 & 0 & s & -c \cr
      s & -c & 0 & 0\end{pmatrix} \,,
    \end{equation}
that generates a distribution in the triangle respects the Parity Token Counting rigidity condition and is hence nonlocal \cite{Boreiri2023}. 
  \item The $B_2$ basis is given by
  \begin{equation}
  M_{B_2}=\frac{1}{2}\begin{pmatrix}
    0 & 0 & \ii\sqrt{2} & -\ii\sqrt{2}  \cr
    \sqrt{2} & -\sqrt{2} & 0 & 0  \cr
    1 & 1 & 1 & 1 \cr
    -1 & -1 & 1 & 1\end{pmatrix} \,.
  \end{equation}
  This is the second iso-entangled basis with tangle$=\frac{1}{2}$. It belongs to the Bell family. It can be written in terms of state-vectors as: 
  $M_{B_2}=\{(\ket{1-}\pm \ket{01})/\sqrt{2}, (\ket{1+}\pm \ii\ket{00})/\sqrt{2}\}$. The autocorrelations are
  \begin{eqnarray}
  mm_A=\frac{1}{2}\begin{pmatrix}
    1 & -1 & 0 & 0  \cr
    -1 & 1 & 0 & 0  \cr
    0 & 0 & 1 & -1 \cr
    0 & 0 & -1 & 1\end{pmatrix}\, , \quad &   mm_B=\frac{1}{2}\begin{pmatrix}
      1 & 1 & -1 & -1  \cr
      1 & 1 & -1 & -1  \cr
      -1 & -1 & 1 & 1 \cr
      -1 & -1 & 1 & 1\end{pmatrix} \,,
  \end{eqnarray}
  and the Bloch vectors are given by
  \begin{align}
  \vec m_A(1) &= \left(-\frac{1}{\sqrt{2}}, 0, 0\right)  &\vec m_B(1) = \left(-\half, 0, -\half\right) \\ 
  \vec m_A(2) &= \left(\frac{1}{\sqrt{2}}, 0, 0\right)   &\vec m_B(2) = \left(-\half, 0, -\half\right) \\
  \vec m_A(3) &= \left(0,-\frac{1}{\sqrt{2}},  0\right)  &\vec m_B(3) = \left(\half, 0, \half\right) \\ 
  \vec m_A(4) &= \left(0,\frac{1}{\sqrt{2}}, 0\right)  &\vec m_B(4) = \left(\half, 0, \half\right)  
  \end{align}
  Here again, since the Bloch vectors correlations differ from those of the $M_{\rm tBSM}$ solution, the $M_{\rm tBSM}$ and $M_{B_2}$ solutions differ.

\end{itemize}

\section{Heuristic method for finding localizable measurements}\label{app:heuristic}

To numerically search for localisable measurements we introduce a cost function $C(M)$ which is minimized when $M$ satisfies the localization equation $f(\mathbf{t},M)=P_\mathbf{t}$ i.e. when it results in a permutation with phases for each possible teleportation distortion laballed by $\mathbf{t}$. Explicitly, we define the cost function as
\begin{equation}
C(M) = \sum_\mathbf{t} \left|\left|2\left|f(\mathbf{t},M)\right|-J\right|-J\right| \,,
\end{equation}
where $J$ is an all-ones matrix. We then use an explicit parametrization of a generic unitary matrix $M$ with the relevant dimension \cite{Spengler2010} and use Nelder-Mead optimization to find the minimum of $C$. 

\section{Twisted Bell basis in higher dimensions}\label{app:twistedbell}

Here, we study the example of the tBSM in order to illustrate how one can use the dimension-free representations to find joint measurements in higher dimensions. For simplicity, we write the tBSM basis in the following form
\begin{equation}
  M_{\rm tBSM} \sim   \frac{1}{2}\left(
\begin{array}{cccc}
 1 & 1 & 1 & 1 \\
 \ii & \ii & -\ii & -\ii \\
 \ii & -\ii & \ii & -\ii \\
 \ii & -\ii & -\ii & \ii \\
\end{array}
\right) \,,
\end{equation}
for which the representation $\{P_{a_1a_2b}\}$ takes a relatively simple form which can be generalised to dimension $d$:
\begin{align}
P_{0,X,X} &= \openone \otimes X ~, & P_{0,Z,X} &= X \otimes w (-1)^{1/d + 1} X Z ~, \nonumber \\
P_{X,0,X} & = X \otimes w \openone ~, & P_{Z,0,X} &= w Z \otimes X ~, \nonumber\\
P_{0,X,Z} &= \openone \otimes X ~, & P_{0,Z,Z} &= \openone \otimes Z~, \nonumber\\ 
P_{X,0,Z} & = X \otimes  \openone ~, & P_{Z,0,Z} &= w Z\otimes \openone ~, \nonumber
\end{align}    

One can verify that for all $d$, (i) all $P_{a_1, a_2, b}$ are permutation matrices with phases, (ii) all four sets $\{P_{a_1, 0, X} \}_{a_1}$, $\{P_{ 0, a_2, X} \}_{a_2}$, $\{P_{a_1, 0, Z} \}_{a_1}$, $\{P_{ 0, a_2, Z} \}_{a_2}$ satisfy the braiding relations and normalisation conditions, and (iii) the sets $\{P_{a_1, 0, b}\}_{a_1}$ and $\{P_{0, a_2, b}\}_{a_2}$ commute for $b=X,Z$. The corresponding intertwiners, $M_X$ and $M_Z$, given by \eqref{eq:3ebitnullspace}, again satisfy the Weyl-Heisenberg relations after normalisation. Finally, the intertwiner $M$ defines a generalisation of tBSM for arbitrary dimensions.

\section{Three-qubit localizable measurements}\label{app:3qubit}

Using numerical search, we find several three-qubit measurements that can be localized at the second level of the Vaidman scheme. The unitary representation of these measurements are solutions to Eq.~\eqref{eq:2ndlevel3qubit} in the main text. A full bank of numerical solutions we found is available online \cite{compapp}. Here we present two examples with special symmetries.

\begin{equation}
  M_{98} = \frac{1}{2}
  \begin{pmatrix}
    0 & 0 & 0 & 0 & 1 & 1 & 1 & 1 \\
    0 & 0 & 0 & 0 & 1 & -1 & 1 & -1 \\
    0 & 0 & 0 & 0 & 1 & 1 & -1 & -1 \\
    1 & 1 & 1 & 1 & 0 & 0 & 0 & 0 \\
    0 & 0 & 0 & 0 & -\ii & \ii & \ii & -\ii \\
    1 & -1 & 1 & -1 & 0 & 0 & 0 & 0 \\
    \ii & \ii & -\ii & -\ii & 0 & 0 & 0 & 0 \\
    1 & -1 & -1 & 1 & 0 & 0 & 0 & 0
  \end{pmatrix},
  \qquad
  M_{106} = \frac{1}{2}
  \begin{pmatrix}
    0 & 0 & 0 & 0 & -1 & 1 & 1 & -1 \\
    0 & 0 & 0 & 0 & 1 & 1 & 1 & 1 \\
    0 & 0 & 0 & 0 & 1 & 1 & -1 & -1 \\
    1 & 1 & 1 & 1 & 0 & 0 & 0 & 0 \\
    0 & 0 & 0 & 0 & 1 & -1 & 1 & -1 \\
    1 & -1 & -1 & 1 & 0 & 0 & 0 & 0 \\
    1 & -1 & 1 & -1 & 0 & 0 & 0 & 0 \\
    -1 & -1 & 1 & 1 & 0 & 0 & 0 & 0
  \end{pmatrix}.
\end{equation}

Both bases share defining properties with the EJM: they are iso-entangled, and the reduced single-party states point to the vertices of simple geometric objects: for $M_{98}$ the single-party reduced states point to the vertices of tetrahedra (different for the different parties), and in the case of $M_{106}$ they point to the vertices of the same square.

\end{document}